\documentclass[reqno]{amsart}
\usepackage{amsmath,amssymb}
\usepackage{graphics,epsfig,color}
\usepackage[mathscr]{euscript}
\usepackage{mathtools} 
\newtheorem{theorem}{Theorem}[section]
\newtheorem{proposition}[theorem]{Proposition}
\newtheorem{lemma}[theorem]{Lemma}

\theoremstyle{definition}
\newtheorem{definition}[theorem]{Definition}
\newtheorem{assumption}[theorem]{Assumption}
\newtheorem{remark}[theorem]{Remark}
\numberwithin{equation}{section}
\numberwithin{theorem}{section}
\renewcommand{\epsilon}{\varepsilon}
\newcommand{\ve}{\varepsilon}
\newcommand{\mc}[1]{{\mathcal #1}}
\newcommand{\mf}[1]{{\mathfrak #1}}
\newcommand{\bb}[1]{{\mathbb #1}}
\newcommand{\ms}[1]{{\mathscr #1}}
\newcommand{\bs}[1]{{\boldsymbol #1}}
\newcommand{\R}{\mathbb{R}}
\newcommand{\eps}{\varepsilon}

\newcommand{\ee}{\mathrm{e}}
\newcommand{\id}{{1 \mskip -5mu {\rm I}}}
\newcommand{\opid}{\mathop  {\rm id}\nolimits}
\newcommand{\Ent}{\mathop{\rm Ent}\nolimits}

\newcommand{\de}{\mathop{}\!\mathrm{d}}

\title[LDP  associated to the homogeneous Boltzmann equation]
{Asymptotic  probability of energy increasing
  solutions to the homogeneous Boltzmann equation}
\author[G.\ Basile]{Giada Basile}
\address{Giada Basile \hfill\break \indent
   Dipartimento di Matematica, Universit\`a di Roma `La Sapienza'
   \hfill\break \indent
   P.le Aldo Moro 2, 00185 Roma, Italy}
 \email{basile@mat.uniroma1.it}
 \author[D.\ Benedetto]{Dario Benedetto}
 \address{Dario Benedetto \hfill\break \indent
   Dipartimento di Matematica, Universit\`a di Roma `La Sapienza'
   \hfill\break \indent
   P.le Aldo Moro 2, 00185 Roma, Italy}
 \email{benedetto@mat.uniroma1.it}
\author[L.\ Bertini]{Lorenzo Bertini}
\address{Lorenzo Bertini \hfill\break \indent
   Dipartimento di Matematica, Universit\`a di Roma `La Sapienza'
   \hfill\break \indent
   P.le Aldo Moro 2, 00185 Roma, Italy}
 \email{bertini@mat.uniroma1.it}
  \author[E.\ Caglioti]{Emanuele Caglioti}
 \address{Emanuele Caglioti \hfill\break \indent
	  Dipartimento di Matematica, Universit\`a di Roma `La Sapienza' 
	\hfill\break \indent
	P.le Aldo Moro 2, 00185 Roma, Italy}
 \email{caglioti@mat.uniroma1.it}

\begin{document}
\begin{abstract}
  Weak solutions to the homogeneous Boltzmann equation with increasing
  energy have been constructed by Lu and Wennberg.  We consider an
  underlying microscopic stochastic model with binary collisions
  (Kac's model) and show that these solutions are atypical. More
  precisely, we prove that the probability of observing these paths is
  exponentially small in the number of particles and compute the
  exponential rate.  This result is obtained by improving the
  established large deviation estimates in the canonical setting. 
  % and proving the large deviation asymptotic 
  % for Lu and Wennberg solutions
  Key ingredients are  the extension of Sanov's theorem
  to the microcanonical ensemble and large deviations for the Kac's
  model in the microcanonical setting.
\end{abstract}


\maketitle
\thispagestyle{empty}

\section{Introduction}
\label{sez:0}   

The derivation of the Boltzmann equation from an underlying
microscopic dynamics of $N$ interacting
particles is a paradigmatic problem in non-equilibrium
statistical mechanics.
%, as it encodes most of the conceptual and technical
%issues.
%This derivation
It is based on the validity of the \emph{Stosszahlansatz}
with probability one 
in the limit $N\to +\infty$.
At a more refined level,
it is possible to 
analyze the corresponding large deviations, whose
derivation is related to the validity
of  the \emph{Stosszahlansatz} with 
probability super-exponentially close to one for $N$ large.

In this perspective, 
the most challenging
case of Newtonian dynamics
of hard spheres in the
Boltzmann-Grad limit has
been recently 
% heuristically % \cite{Bou} and rigorously in
discussed in
\cite{BGSS2}.
Nevertheless, also the case of stochastic dynamics presents interesting features.
The first result in this setting has been  obtained in \cite{Le}, where
a large deviation upper bound is derived in the space homogeneous
case. A complete large deviation principle has been
obtained in \cite{Re}
for a space inhomogeneous model with a finite set of  velocities.
In \cite{BBBO}
a large deviation upper bound is achieved for 
 a homogeneous model which conserves momentum but not
energy, 
while the matching lower bound is obtained for a restricted class of paths.
A similar  result, in the case of energy and momentum conservation, has been  proven
in \cite{He}. In this case the upper and lower  bound match for a subset of paths for which energy is conserved.

For energy preserving microscopic dynamics with unbounded velocities, a main obstacle to a complete proof of large deviations 
is the 
occurrence of macroscopic paths with finite rate function
that violate the conservation of the energy.  In particular,
as discussed in  \cite{He},
a class of such paths is given by the solutions to the
homogeneous Boltzmann equations constructed by Lu and Wennberg in
\cite{LuW}, for which the energy is increasing.
Another example of large deviation asymptotic for non-conserving
energy path has been constructed in \cite{BBBC},
for a Kac-like microscopic dynamics with discrete energies.
More precisely, as proven in \cite{He}, the upper bound rate function derived
in \cite{Le} vanishes on Lu and Wennberg solutions,
while their asymptotic probability is 
$\ee^{-cN}$, which implies the upper bound rate function
in \cite{Le} is not optimal.

The homogeneous Boltzmann equation with hard sphere  cross-section
reads as
\begin{equation}
  \label{eq:hb}
  \partial_t f_t(v) =
  \frac 12\int_{\bb R^d}  \de v_* \! \int_{S_{d-1}} \de \omega \,
  |(v-v_*)\cdot \omega| \big( f_t(v') f_t(v_*')  -
  f_t(v) f_t(v_*) \big).
\end{equation}  
where $S_{d-1}$ is the sphere in $\bb R^d$.
The associate Cauchy problem has a unique solution in the class
of function with constant energy \cite{MW}. 
Let us discuss how the Lu and Wennberg solutions, with increasing
energy, can be constructed
in the special case in which the energy has a unique jump at time zero.
Consider a sequence of initial densities $f_0^n$  such that
$f_0^n$ converges weakly to $f_0$ but
$e \coloneqq \lim_{n} \int f_0^n(v) v^2 \de v > \int f_0 (v) v^2 \de v$,
namely a fraction of energy evaporates at infinity.
Denoting by $f_t^n$ the unique energy conserving solution
of the homogeneous Boltzmann equation
with initial datum $f_0^n$,
we then have that $f^n_t$, $t\ge 0$,  converges to
a solution to the homogeneous Boltzmann equation,
with initial datum $f_0$, but the energy has a positive jump
at time $0$.
%Indeed, by the estimates in \cite{MW,W}  for $t>0$, $\int f^n_t(v) v^{2+\beta} \de v$ is bounded uniformly in $n$, so that$e = \lim_{n} \int f^n_t(v) v^2 \de v = \int f_t(v) v^2 \de v$.
Observe that this construction does not yield a jump
in the energy if the total cross section
is bounded, in fact in this case there is uniqueness of the
solution without the requirement of energy conservation.
A model with this feature has been analyzed in \cite{Er2}.

A main result of this paper is the proposal of
a rate function that improves the one in \cite{Le}, 
being strictly positive on Lu and Wennberg solutions.
In particular we consider a Kac walk with the hard sphere cross section,
and prove the large deviation upper bound with such rate function.
The matching large deviation lower bound is achieved for both
Lu and Wennberg solutions and the same restricted
class of path as in \cite{BBBO,He}.

As it is clear from the previous construction, Lu and Wennberg solutions
can be produced from a microscopic model only 
if there exists a fluctuation of the initial energy
and then following the typical behavior.
To introduce the improved rate function
we consider first the case in which the initial
velocities are sampled from the microcanonical ensembles,
namely the total energy and momentum are not random, and
given by $(eN,uN)$.
After \cite{BBBC,BBBO}, we consider as empirical observable 
the pair $(\pi^N,Q^N)$ where $\pi^N$ is the empirical distribution of
velocities, while the empirical flux $Q^N$ records the collision
times together the incoming and outgoing velocities.
The microcanonical rate function  reads
\begin{equation}
  \label{IHJ}
  I_{e,u}(\pi,Q) = H_{e,u}(\pi_0) + J_{e,u}(\pi,Q),
\end{equation}
where $H_{e,u}$ takes into account the fluctuation of the
initial data, while $J_{e,u}$ is the dynamical contribution, that is defined as follows.
Set $\de Q^\pi\coloneqq
\frac 1 2  \de \pi\otimes \de \pi\,  B \de\omega \de t$, with $B=B(v-v_*, \omega)=\frac 1 2 |(v-v_*)\cdot\omega|$,
and let $J(\pi, Q)$ be  the relative entropy of $Q$ with respect to $Q^\pi$, namely 
\begin{equation}\label{Jint}
J(\pi, Q)=\int\Big\{ \de Q \log \frac{\de Q}{\de Q^\pi} -\de Q +\de Q^\pi\Big\}.
\end{equation}
Then, by the microcanonical constraint,
$J_{e,u}(\pi, Q)$ is equal to $J(\pi, Q)$ if the energy of $\pi$ does not exceed $e$ and its momentum is equal to $u$,
while $J_{e,u}(\pi, Q)=+\infty$ otherwise.
The functional $H_{e,u}$
will be derived by extending Sanov's theorem to
the microcanonical  ensemble. In particular, $H_{e,u}(\pi_0)$
is infinite when the energy of $\pi_0$ exceed  $e$,
but it can be finite when the energy is below $e$.
Namely, loss of energy at time $0$ occurs with exponentially
small probability.
%Since the microscopic dynamics conserves the energy,
%$J_{e,u}(\pi,Q)<+\infty$ implies that the energy at any time 
%does not exceed $e$.
According to \eqref{IHJ},
the  asymptotic probability of Lu and Wennberg solutions
is then $\exp (-NH_{e,u}(\pi_0))$.

We then analyze the case in which the initial velocities are sampled
from the canonical ensemble, namely are i.i.d. $m$-distributed  random variables.
The canonical rate function can then be obtained from
\eqref{IHJ} as follows
\begin{equation}\label{Iint}
  I(\pi,Q) = \inf_{e,u} \big( A(e,u) + I_{e,u} (\pi,Q) \big),
\end{equation}
where $A$ %, as defined in \eqref{cra},
is the rate function
for the energy and momentum of  the sum of i.i.d. $m$-distributed random variables,
given by Cram\'er's theorem.
The rate function introduced in \cite{Le} and further analyzed in \cite{He} is given by
$$\mc I (\pi, Q)= \Ent (\pi_0|m)+ J(\pi, Q),$$
where $\Ent(\pi_0|m)$ is the relative entropy. In particular, $\mc I$ vanishes on  Lu and Wennberg solutions. We show that 
$I$ defined in \eqref{Iint} is larger than $\mc I$ and vanishes only on the unique energy conserving solution to
\eqref{eq:hb}. Moreover, we compute explicitly its value 
on the Lu and Wennberg solutions, which is given by $c \Delta \mc E$,
where $c$ is a strictly positive constant depending on the tail of initial distribution  $m$
and $\Delta \mc E$ is the total gain of the energy.
Hence, the asymptotic probability of Lu and Wennberg solutions is 
$\ee^{-cN \Delta \mc E}$.

The present work is organized as follows.
In Section \ref{sez:1} we consider the static case,
by analyzing the large deviations of the empirical measure
when the velocities are sampled from the microcanonical ensemble.
As discussed before, we show that the large deviation functional
is finite on probability measures with energy evaporation.
In Section \ref{sez:2}
we state the large deviation principle for the Kac model with  hard sphere cross section
and microcanonical initial data. The corresponding
proof is carried out in Sections \ref{sez:3}, \ref{sez:4}.
In Section \ref{sez:6} we derive the large deviation
asymptotic for the Kac model with canonical
initial distribution.
Section \ref{sez:5} is finally devoted to the asymptotic probability  of
Lu and Wennberg solutions.

\section{Sanov theorem for microcanonical ensemble}
\label{sez:1}

Sanov's theorem, that describes the asymptotic behavior of the empirical
measures associated to a sequence of $N$ i.i.d.\ random variables, is a
basic result in the theory of large deviations. A natural question is
to replace the independence assumption by some dependency structure.
For instance, the case of the empirical measure associated to Markov
chains is the content of the classical Donsker-Varadhan theorem.
%while
%the behavior of the empirical spectral measure for random matrices
%has been analyzed in \cite{BG}.
We here analyze
the case in which the
underlying sequence of random variables is sampled according to a
microcanonical ensemble, that can be realized by conditioning i.i.d.\ 
random variables to the sum of their squares, i.e.\ to the total
kinetic energy in physical interpretation.
A particular case of this
situation has been previously discussed in \cite{KR};
there  it is in fact analyzed the
case where $N$ real random variables are sampled according to the
uniform measure on the sphere of radius $\sqrt{N}$ on $\bb R^{N}$ and
corresponds to the microcanonical ensemble associated to i.i.d.\
Gaussians. A peculiar feature of this setting is the possibility of
observing -- at the large deviations level -- probabilities that violate
the microcanonical constraint. More precisely, while for each $N$ the
law of the  empirical measure is supported by the probabilities with
fixed second moment, the large deviations rate function is finite also
for probabilities with second moment strictly smaller than the
prescribed value. In view of the application to homogeneous Boltzmann
equations, we shall next consider microcanonical ensembles that are
obtained by conditioning both to the total energy and to the total
momentum.

Fix hereafter $d\ge 2$ and denote by $\ms P(\bb R^d)$ the set of
probability measures on $\bb R^d$ equipped with the topology induced
by the weak convergence and the associated Borel $\sigma$-algebra.
%{\color{red} Stroger topology? Wasserstien $p< 2$?}
Let $\bs \zeta\colon \bb R^d\mapsto [0, +\infty)\times \bb R^d$ be the map
given by $\bs \zeta =(\zeta_0, \zeta)(v)=(|v|^2/2, v)$.  We shall
consider probabilities $m\in \ms P(\bb R^d)$ satisfying the following
conditions.

\begin{assumption}\label{ass:2}
  There exists $\gamma_0^*\in (0, +\infty]$ such that
  \begin{itemize}
  \item[(i)] $m$ is absolutely continuous with respect to the Lebesgue
    measure and $m$ is strictly positive on open sets;
  \item[(ii)] $m(\ee^{\gamma_0 \zeta_0}) < +\infty$ for any
    $\gamma_0\in (-\infty, \gamma_0^*)$, and
    $\lim_{\gamma_0\uparrow \gamma_0^*} m(\ee^{\gamma_0
      \zeta_0})=+\infty$;
  \item[(iii)] for each
    $\bs \gamma=(\gamma_0, \gamma)\in (-\infty, \gamma_0^*)\times \bb
    R^d$ the Fourier transform of
    $\frac {\de m}{\de v}\ee^{\bs\gamma \cdot \bs\zeta}$ belongs to
    $L^1(\bb R^d)$;
  \item[(iv)]
    there exists $c>0$ such that
    $\frac {\de m}{\de v} \ge \frac 1c\exp\{-c  |v|^2\}$.
  \end{itemize}
\end{assumption}
Condition (iv) is mainly technical, and will be used only to derive
the lower bound for the dynamical rate function.

We observe that the map
$(-\infty,\gamma_0^*)\times\bb R^d\ni\bs\gamma\mapsto \log
m(\ee^{\bs\gamma\cdot\bs\zeta})$ is strictly convex.
Set  $Z= \{ (e,u)\in (0,+\infty) \times \bb R^d:\,
e > |u|^2/2\}$, then 
$\nabla \log m(\ee^{\bs\gamma\cdot\bs\zeta})$ is a bijection from
$ (-\infty, \gamma_0^*)\times \bb R^d$ to $Z$.
We denote by
$(e, u)\mapsto \bs \gamma(e,u)$ the inverse map and by $m_{e,u}$ the
probability on $\bb R^d$ defined by
\begin{equation}\label{meu}
  m_{e,u} (\de v)
  \coloneqq \frac {\ee^{\bs\gamma(e,u)\cdot\bs\zeta (v)}}
  {m(\ee^{\bs\gamma(e,u)\cdot\bs\zeta})} 
  m(\de v).
\end{equation}
In words, $m_{e,u}$ is the \emph{exponential tilt} of $m$ such that
$m_{e,u}(\bs \zeta) =(e,u)$.
Namely, $u$ and $e$ are the average values of velocity and 
total energy, respectively. 
Note that $m_{e,u}$ satisfies the
conditions in Assumption~\ref{ass:2} with $\gamma_0^*$ replaced by
$\gamma_0^*-\gamma_0(e,u)$.
We denote by $U$ the internal
energy defined by the relation $e=U+|u|^2/2$, so that
$U$ is the expected value
of $|v-u|^2/2$.

Let $\Sigma^N\coloneqq \big(\bb R^d\big)^N$ be the configuration space for $N$
velocities in $\bb R^d$. Given $(e,u)\in Z$,
we denote by
\begin{equation}\label{sig}
  \Sigma^N_{e,u}\coloneqq \Big\{\bs v \in (\bb R^d)^N :\,
  \frac 1 N\sum_{i=1}^N \bs \zeta(v_i) = (e,u)
  \Big\} 
\end{equation}
the set of configurations with total momentum $N u$ and total energy
$N e$.

Let $\mu^N$,
be the
probability on $\Sigma^N$ given by $m^{\otimes N}$,
interpreted as the canonical ensemble.
Let also
$(e,u)\mapsto\nu^N_{e,u}$ be a regular version of the probability
$\mu^N$ conditioned to
$\frac 1 N\sum_{i=1}^N \bs \zeta(v_i)$.
In particular, $\nu^N_{e,u}$, interpreted as the
microcanonical ensemble, is the probability supported by
$\Sigma_{e,u}^N$ informally given by
$\nu^N_{e,u} = \mu^{N}(\ \cdot\ \vert \Sigma^N_{e,u} )$.
%It is then simple to check that
As $N\to \infty$ the one-marginal of
$\{\nu^N_{e,u}\}$ converge to $m_{e,u}$ (equivalence of ensembles in
the thermodynamic limit),
see 
\cite[\S 1.5]{CCL-RLV}
and \cite{Nam}.
Our aim is to describe the corresponding
large deviations asymptotic. In order to apply this result to Kac's
walk with a canonical initial distribution of the velocities, the
large deviation principle will be proven uniformly for $(e,u)$ in
compact subsets of $Z$.

We define the \emph{empirical measure} as the map
$\pi^N \colon \Sigma^N \to \ms P(\bb R^d)$ given by
\begin{equation}
  \label{1}
  \pi^N(\bs v)=\frac 1 N \sum_i \delta_{v_i}.
\end{equation}
Given two probabilities $m_1, m_2$, recall that the relative entropy
$\Ent(m_2\vert m_1)$ is defined as
$\Ent(m_2\vert m_1)=\int \de m_1 \rho\log\rho$, where
$\de m_2=\rho\, \de m_1$, understanding that
$\Ent(m_2|m_1)=+\infty$ if $m_2$ is not absolutely continuous
with respect to $m_1$.

Given $(e,u)\in Z$ set
\begin{equation}
  \label{ceu}
  C_{e,u} \coloneqq\big\{\pi\in \ms P(\bb R^d) \colon \pi(\zeta)=u,\,
  \pi(\zeta_0)\le e\big\}
\end{equation}
that is a compact and convex subset of $\ms P(\bb R^d)$. Note 
that $ C_{e,u}$ is the closure in $\ms P(\bb R^d)$ of the set of
probabilities $\pi$ satisfying the microcanonical constraint
$\pi(\bs \zeta) =(e,u)$.

\begin{theorem}\label{sld}
  Fix $(e,u)\in Z$ and a sequence
  $(e_N,u_N)\to (e,u)$.
  If $m$ satisfies item (i)--(iii)
  in Assumption \ref{ass:2} then the family of probabilities
  $\{ \nu^N_{e_N,u_N}\circ (\pi^N)^{-1}\}$ on $\ms P(\bb R^d)$
  satisfies a large deviation principle with good and convex rate
  function $H_{e,u}\colon \ms P(\bb R^d)\to [0, +\infty]$ given by
  \begin{equation}
    \label{Heu}
      H_{e,u}(\pi)= \begin{cases}
      \Ent(\pi\vert m_{e,u} ) + \big[\gamma_0^* -
      \gamma_0(e,u)\big]\, \big[e-\pi\big(\zeta_0\big)\big] & \textrm{if }\,
      \pi \in C_{e,u},\\
      %\pi\big(\zeta_0\big)\leq e, \; \pi(\zeta)=u\\
        +\infty & \textrm{otherwise}.
      \end{cases}
    \end{equation}  
\end{theorem}  
The rate function $H_{e,u}$ can be understood
as the canonical rate function $\Ent(\,\cdot\, \vert m_{e,u})$
with an extra penalization for violations of the energy constraint.
When $m$ is the standard Gaussian on $\bb R$ and the
momentum constraint is dropped this result reduces to the one
obtained in \cite{KR}. 

While the arguments in \cite{KR} rely on the representation of the
uniform measure on the spheres in terms of i.i.d.\ Gaussian, 
the proof of the above theorem will be achieved by applying the
G\"artner-Ellis theorem, which provides the large deviation
rate function as the Legendre transform of the log-moment generating
function.
To this end, for  $\phi\in C_{\mathrm b}(\bb R^d)$ set
\begin{equation}\label{meuphi}
  \de m_{e,u}^\phi \coloneqq \frac {\de m_{e,u} \ee^\phi}{m_{e,u}(\ee^\phi)},
\end{equation}
and
\begin{equation}
  \Lambda_{e,u}(\phi) \coloneqq -\bs \gamma \cdot (e,u)
  + \log m_{e,u}(\ee^{\phi +\bs\gamma\cdot\zeta}), 
\end{equation}
where $\bs\gamma=\bs \gamma(\phi)$ is chosen so that
\begin{equation}\label{gammaphi}
  \frac{m_{e,u}^{\phi}\big(\exp\{\bs\gamma\cdot\bs\zeta\}\,\bs\zeta\big)}{m_{e,u}^{\phi}\big(\exp\{\bs\gamma\cdot\bs\zeta\}\big)}=(e,u),
  \end{equation}
  namely, it is chosen in order that the exponential tilt of $m^\phi$
  has average energy and momentum
  $(e,u)$.
  
  \begin{lemma}\label{lemma1}
    For each $\phi\in C_{\mathrm b}(\bb R^d)$,
    \begin{equation}
      \lim_{N\to\infty}\frac 1 N\log
      \nu^N_{e_N,u_N}\big(\ee^{N\pi^N(\phi)}\big)=\Lambda_{e,u}(\phi). 
    \end{equation}
  \end{lemma}

  \begin{proof}
    As simple to check,
    \begin{equation*}
      \Big|
      \frac 1 N\log \nu^N_{e_N,u_N}\big(\ee^{N\pi^N(\phi_2)}\big)
      -\frac 1 N\log \nu^N_{e_N,u_N}\big(\ee^{N\pi^N(\phi_1)}\big)
      \Big| \le  \sup_{v\in \bb R^d} |\phi_2(v) -\phi_1(v)|.
    \end{equation*}
    By a density argument, it is therefore enough to prove the
    statement for smooth $\phi$.

    Observing that
    $m^{\otimes N} (\cdot\vert \Sigma_{e,u}^N )=m_{e,u}^{\otimes N}
    (\cdot\vert \Sigma_{e,u}^N )$, for
    $\bs\gamma \in (-\infty,\gamma_0^*)\times \bb R^d$ we write
  \begin{equation*}
    \frac 1 N\log \nu^N_{e_N,u_N}\big(\ee^{N\pi^N(\phi)}\big)=
     - \bs \gamma \cdot (e_N,u_N) 
     +\frac 1 N \log m_{e_N,u_N}^{\otimes N}\big(
    \ee^{N\pi^N(\phi+\bs\gamma\cdot\bs\zeta)}\big). 
  \end{equation*} 
  By a direct computation (cfr.\ Lemma~3.5 in \cite{BBBO}) for any
  $\psi\in C_{\mathrm b}(\bb R^d)$ 
  \begin{equation*}
    m_{e_N,u_N}^{\otimes N}\big( \ee^{N\pi^N(\psi)}\vert \Sigma^N_{e_N,u_N}\big)
    = \big(m_{e_N,u_N}(\ee^\psi)\big)^N \frac {f_N^\psi(e_N,u_N)}{f_N(e_N,u_N)}, 
  \end{equation*}  
  where $f_N^\psi$, $f_N$ are the densities of the random vector
  $\frac 1 N\sum_{i=1}^N \bs\zeta(v_i)$, in which $\{v_i\}$ are
  i.i.d.\ with law $m_{e_N,u_N}^\psi$, $m_{e_N,u_N}$ respectively.
  Observe that, as we assumed that $m$ is strictly positive on open set,
  the law of $\frac 1N \sum_i \zeta(v_i)$ is absolutely continuous
  for $N>2$.
  Choosing
  $\psi=\phi+\bs\gamma\cdot\bs\zeta$, with $\bs\gamma=\bs\gamma_N(\phi)$
  such that \eqref{gammaphi} holds
  with $(e,u)$ replaced by $(e_N,u_N)$, by the local central limit theorem
  (see e.g. \cite{Pe})
  we deduce
  \begin{equation*}
    \lim_{N\to\infty} \frac 1 N\log \frac
    {f_N^{\phi+\bs\gamma\cdot\bs\zeta}(e_N,u_N)}{f_N(e_N,u_N)}=0.      
  \end{equation*}
  Note indeed that the local central limit holds in view of Assumption
  \ref{ass:2} and the smoothness of $\phi$.  Gathering the above
  computations, the statement follows.
\end{proof}  

\begin{lemma}\label{lemma2}
  Let
  $\Lambda_{e,u}^*(\pi)\coloneqq \sup_{\phi}\{\pi(\phi)-\Lambda_{e,u}(\phi)\}$
  be the Legendre transform of $\Lambda_{e,u}$. Then
  $\Lambda_{e,u}^*=H_{e,u}$.
\end{lemma}

\begin{proof}
  Note that $\Lambda^*_{e,u}(\pi)<+\infty$ implies
  $\pi(\zeta_0)<+\infty$. For such $\pi$, recalling \eqref{gammaphi},
  \begin{equation*}\begin{split}
      &\Lambda^*_{e,u}(\pi) =\sup_{\phi}\{\pi(\phi) +
      \bs \gamma \cdot (e,u)
      - \log m_{e,u}(\ee^{\phi +\bs\gamma\cdot\zeta})
      \}\\
      &=\sup_{\bs\gamma:\gamma_0<\gamma_0^*-\gamma_0(e,u)}\sup_{\phi}\{\pi(\phi)
      +
      \bs \gamma \cdot (e,u)
      - \log m_{e,u}(\ee^{\phi
        +\bs\gamma\cdot\zeta})
      \}\\
      &=\sup_{\bs\gamma:\gamma_0<\gamma_0^*-\gamma_0(e,u)}\sup_{\phi}\{\pi(\phi+\bs\gamma\cdot\bs\zeta)
      +\bs \gamma \cdot\big(e-\pi(\zeta_0), u -\pi(\zeta)\big) - \log
      m_{e,u}(\ee^{\phi +\bs\gamma\cdot\zeta})
      \}\\
      &=\Ent(\pi\vert m_{e,u})+
      \sup_{\bs\gamma:\gamma_0<\gamma_0^*-\gamma_0(e,u)}\{ \bs \gamma
      \cdot\big(e-\pi(\zeta_0), u -\pi(\zeta)\big) \}=H_{e,u}(\pi)
    \end{split}\end{equation*}
  that concludes the proof.
\end{proof}  

\begin{proof}[Proof of Theorem \ref{sld}]
  For $\delta>0$ let $C^\delta_{e,u}$ be the compact subset of
  $\ms P(\bb R^d)$ given by
  $C^\delta_{e,u} \coloneqq \{\pi \in \ms P(\bb \R^d) \colon
  \pi(\zeta_0)\le e + \delta,\, |\pi(\zeta) - u|\le \delta \}$.  By
  the very definition of $\nu_{e,u}^N$, definitely in $N$ we have
  $\nu_{e_N,u_N}^N(\pi^N\in C^\delta_{e,u} )=1$, which implies the
  exponential tightness of the sequence
  $\nu_{e_N,u_N}^N\circ (\pi^N)^{-1}$.

  Since the map $\pi\mapsto H_{e,u}(\pi)$ is strictly convex, in the
  terminology of convex analysis used in
  \cite[Thm.~4.5.20]{DZ},
  every $\pi\in C_{e,u}$ is an \emph{exposed} point
  of $H_{e,u}$. Therefore the 
  statement follows from Lemmata \ref{lemma1} and \ref{lemma2} by the
  abstract G\"artner-Ellis theorem.
\end{proof}

\subsection*{Large deviations from total probability formula}

We next show how the Sanov's theorem for i.i.d.\
random variables can be recovered from Theorem \ref{sld}. While this
route is overcomplicated in the present context, it will be crucial to
deduce the large deviations for Kac's walks with canonical initial
distribution of the velocities.

We first state a general argument to deduce the large deviation principle from
the total probability formula.
Let $\ms X$ be a Hausdorff topological space and $\mu^n$ be a sequence
of probabilities on $\ms X$. Let also $\ms Y$ be a locally compact Polish
space, $Y$ be a $\ms Y$-valued random variable on $\ms X$ and
denote by $p_n$ its law.  Letting $y\mapsto \nu^n_y$ be a regular
version of the conditional probability of $\mu^n$ given $Y$ we have
the disintegration
\begin{equation}
  \label{dis}
\mu^n=\int \!p_n(\de y)\,\nu^n_y.
\end{equation}  
We will deduce the large deviation of $\mu^n$ from the large
deviations of $p_n$ and the large deviations on $\nu^n_y$, that will
be assumed to hold uniformly for $y$ in compact subsets of $\ms Y$.

\begin{proposition}
  \label{t:ldtp}
  Assume: 
  \begin{itemize}
  \item[(i)]
    the family $\{p_n\}$ is exponentially tight and satisfies a large
    deviation principle with good rate function
    $A\colon \ms Y \to [0,+\infty]$;
  \item[(ii)]
    for each compact $K\subset\subset \ms Y$ there exists a sequence
    of compacts $H_\ell\subset \subset \ms X$ such that
    $\sup_{y\in K} \nu_y^n\big(H_\ell^\mathrm{c}\big) \leq
    \ee^{-n\ell}$;
  \item[(iii)]
    for each $y\in \ms Y$ and each sequence $y_n\to y$ the
    family $\{\nu^n_{y_n}\}$ satisfies a large deviation principle with
    good rate function $F_y\colon \ms X \to [0,+\infty]$.
  \end{itemize}
  Then the family $\{\mu^n\}$
  is exponentially tight and satisfies a large deviation
  principle with good rate function $I\colon \ms X \to [0,+\infty]$ given by
  \begin{equation}
    \label{I=}
  I(x)=\inf_{y\in\ms Y}\big\{ A(y) + F_y(x)\big\}.
  \end{equation}
\end{proposition}

\begin{proof}$~$
  
  \smallskip\noindent\emph{Step 1. Exponential tightness}.
  As follows from \eqref{dis}, for each compact $K\subset\subset \ms Y$ and each
  compact $H\subset\subset \ms X$
  \begin{equation*}
    \mu^n(H^\mathrm{c})\leq \sup_{y\in K} \nu^n_y(H^\mathrm{c}) + p_n
    (K^\mathrm{c}). 
  \end{equation*}  
  The assumptions on $\{p_n\}$ and $\{\nu_y^n\}$ thus yield
  the exponential tightness of $\{\mu^n\}$.

  \smallskip\noindent\emph{Step 2. Lower semicontinuity of the rate function.}
  Since $A$ is lower semicontinuous, 
  the lower semicontinuity of $I$ in \eqref{I=} is implied by the
  (joint) lower semicontinuity of the map
  $\ms X\times \ms Y \ni (x, y)\mapsto F_y(x)$ that we next deduce.
  Since $\ms Y$ is Polish, the joint lower semicontinuity of $F$ is in
  fact equivalent to the following statement.  For each
  $(x,y)\in \ms X\times \ms Y$, each sequence $y_k\to y$, and each
  $\delta>0$ there exists an open neighborhood $\mathcal N \ni x$
  such that
  \begin{equation}
    \label{jlsc}
    \varliminf_{k} \inf_{x' \in \mc N } F_{y_k} (x') \ge F_y(x) - \delta.
  \end{equation}
  Fix $(x,y)\in \ms X\times \ms Y$, a sequence $y_k\to y$, and
  $\delta>0$.  By the lower semicontinuity of
  $\ms X \ni x \mapsto F_y(x)$, there exists an open neighborhood
  ${\mathcal N}' \ni x$ such that
  \begin{equation}
    \label{lsc1}
    \inf_{x' \in \mc N' } F_{y} (x') \ge F_y(x) - \delta.
  \end{equation}
  Denoting by an over-line the closure, let now $\mathcal N$
  be a open neighborhood such that
  $x\in {\mathcal N} \subset \overline{{\mathcal N}} \subset{\mathcal
    N}'$. We then claim that the bound \eqref{jlsc} holds. In order to
  show it, by passing to a not relabeled subsequence, we may assume
  that
  $\varliminf_{k} \inf_{x' \in \mc N } F_{y_k} (x')= \lim_{k} \inf_{x'
    \in \mc N } F_{y_k} (x')$.  For $k$ fixed, by the lower bound for
  the sequence $\{\nu^n_{y_k}\}$,
  \begin{equation*}
    \varliminf_{n} \frac 1n \log \nu^n_{y_k} \big(\mathcal N \big) \ge  -
    \inf_{x'\in\mathcal N}  F_{y_k}(x') 
  \end{equation*}
  which, by taking the inferior limit in $k$, implies 
  \begin{equation*}
    \varliminf_k\varliminf_{n} \frac 1n
    \log \nu^n_{y_k} \big(\mathcal N \big) \ge  -
    \lim_k \inf_{x'\in\mathcal N}  F_{y_k}(x').
  \end{equation*}
  By a diagonal argument, there exists a sequence $n_k\uparrow +\infty$ such that
  \begin{equation*}
    \begin{split}
      &\varliminf_k\varliminf_{n} \frac 1n \log \nu^n_{y_k}
      \big(\mathcal N \big) 
      = \varliminf_k\frac 1{n_k} \log \nu^{n_k}_{y_k} \big(\mathcal N \big)
      \le \varlimsup_k\frac 1{n_k} \log \nu^{n_k}_{y_k}
      \big(\overline{\mathcal N} \big)
      \\
      &\qquad \le - \inf_{x' \in \overline{\mathcal N}} F_y(x')
      \le - \inf_{x' \in {\mathcal N}'} F_y(x')
      \le -   \big[ F_y(x) -\delta\big]   
    \end{split}
  \end{equation*}
  where we used the large deviations upper bound for the sequence
  $\{\nu^{n_k}_{y_k}\}$ and \eqref{lsc1} in the last step.
  Comparing the two last displayed equations the bound \eqref{jlsc} follows.

  \smallskip\noindent\emph{Step 3. Lower bound.}
  It is enough to show that for each $x\in \ms X$ and each  open
  neighborhood $\mathcal N \ni x$
  \begin{equation}
    \label{enlb}
    \varliminf_n \frac 1n \log \mu^n\big(\mathcal N\big) \ge - I(x).
  \end{equation}

  Fix a metric inducing the topology of $\ms Y$ and, given $y\in \ms
  Y$ and $\delta>0$, let $B_\delta(y)$ the corresponding open ball of
  radius $\delta$ centered in $y$. 
  In order to show \eqref{enlb}, fix $y \in \ms Y$.
  By the large deviations lower bound of the sequence $\{p_n\}$,
  for each $\delta>0$ we then have
  \begin{equation*}
    \varliminf_n \frac 1n \log p_n\big(B_\delta(y) \big) \ge - A(y).
  \end{equation*}
  Therefore, by a diagonal argument, there exists a sequence
  $\delta_n\downarrow 0$ such that 
  \begin{equation*}
    \varliminf_n \frac 1n \log p_n\big(B_{\delta_n}(y) \big) \ge - A(y).
  \end{equation*}
  From the disintegration \eqref{dis} we then obtain
  \begin{equation*}
    \mu^n\big(\mathcal N\big) \ge
    \int_{B_{\delta_n}(y)} \!p_n(\de y')\, \nu^n_{y'}\big(\mathcal N\big)
    \ge p_n \big(B_{\delta_n}(y) \big) \,
    \inf_{y'\in B_{\delta_n}(y)} \nu^n_{y'}\big(\mathcal N\big).
  \end{equation*}
  Whence, for a suitable sequence $y'_n \to y$,
  \begin{equation*}
    \varliminf_n \frac 1n \log\mu^n\big(\mathcal N\big)
    \ge
    \varliminf_n \frac 1n \log p_n \big(B_{\delta_n}(y) \big)
    + \varliminf_n \frac 1n \log \nu^n_{y'_n}\big(\mathcal N\big)
    \ge - \big[ A(y) + F_y(x) \big]
  \end{equation*}
  where we used the large deviations lower bound for the family 
  $\{ \nu^n_{y'_n}\}$. By optimizing over $y\in\ms Y$ and recalling
  \eqref{I=} we then deduce \eqref{enlb}.

  \smallskip\noindent\emph{Step 4. Upper bound for compacts.}
  Fix a compact set $H \subset\subset \ms X$, $\ell>0$,
  $\varepsilon>0$, and observe that, by the joint lower semicontinuity
  of $F$ proven in Step~2 above, the map
  $\ms Y \ni y \mapsto \inf_{x\in H} F_{y} (x)$ is lower
  semicontinuous. 
  By the exponential tightness of $\{p_n\}$, there exists a compact
  $K_\ell\subset\subset \ms Y$ such that
  $p_n\big(K_\ell^\mathrm{c}\big) \le \ee^{-n\ell}$. For each $y\in K_\ell$,
  by the lower semicontinuity of $A$ and the previous observation,
  there exists $\delta>0$ such 
  that $A(y')\ge A(y) - \varepsilon/2$ and
  $\inf_{x\in H} F_{y'} (x) \ge \inf_{x\in H} F_{y} (x)-\varepsilon/2$
  for any $y'\in {B}_{2 \delta}(y)$.  
  By the local compactness of $\ms Y$, possibly by decreasing
  $\delta$, we can assume that ${B}_{2 \delta}(y)$  is relatively compact.
  Furthermore, by the compactness of $K_\ell$, there exists a finite family
  $\{B_{\delta_i}(y_i)\}_{i=1,\ldots,r}$ such that
  $K_\ell\subset \bigcup_i B_{\delta_i}(y_i)$. In view of \eqref{dis},
  \begin{equation}\label{1.11}
    \begin{split}
    \mu^n\big(H \big) &\le \sum_{i=1}^r
    \int_{ B_{\delta_i}(y_i)} \!p_n(\de y') \,\nu^n_{y'}\big(H\big)
    + p_n\big(K_\ell^\mathrm{c}\big)
    \\
    &\le \sum_{i=1}^r
    p_n\big( \overline{B}_{\delta_i}(y_i)\big)
    \sup_{y'\in  {B}_{\delta_i}(y_i)} \nu^n_{y'}\big(H\big)
    + \ee^{-n\ell}.
    \end{split}
  \end{equation}
  Since the sets ${B}_{\delta_i}(y_i)$ are relatively compact,
  by passing if necessary to a not relabeled subsequence, for each
  $i=1,\ldots,r$ there exist
  $\bar y_i \in \overline{B}_{\delta_i}(y_i)$ and a sequence
  $y^n_i \to \bar y_i$ such that
  \begin{equation*}
    \varlimsup_n \sup_{y'\in {B}_{\delta_i}(y_i)} \frac 1n \log
        \nu^n_{y'}\big(H\big) 
    =\varlimsup_n \frac 1n \log \nu^n_{y^n_i}\big(H\big).   
  \end{equation*}
  Letting $a\vee b \coloneqq \max \{a,b\}$ and using the large deviation
  upper bound both for $\{p_n\}$ and $\{\nu^n_{y^n_i}\}$ in
  \eqref{1.11} we thus get
  \begin{equation*}
    \begin{split}
    \varlimsup_{n}\frac 1n \log \mu^n \big(H\big) 
    & \le
    \max_{i=1,\dots,r} \Big\{ -
    \inf_{y' \in \overline{B}_{\delta_i}(y_i)} A(y') -
    \inf_{x\in H} F_{\bar y_i} (x) \Big\} \vee (-\ell)
    \\
    & \le - \min_{i=1,\dots,r} \Big\{ A(y_i) 
    + \inf_{x\in  H} F_{y_i} (x) -\epsilon \Big\} \vee (-\ell)
    \\&
    \le -\inf_{x\in  H}\, \inf_{y\in\ms Y}
    \big\{ A(y) + F_y(x) -\varepsilon \big\} \vee (-\ell).
  \end{split}
  \end{equation*}
  Recalling \eqref{I=}, we conclude by taking the limits
  $\epsilon\downarrow 0$ and $\ell\uparrow +\infty$.
\end{proof}

Let $\bs v \in \Sigma^N$ be sampled according to the product
probability $\mu^N=m^{\otimes N}$ and denote by $p_N$ the law of
$\frac 1N \sum_i \bs \zeta(v_i)$.  
We then have the disintegration
\begin{equation*}
  \mu^N = \int p_N\big(\de (e,u)\big) \, \nu^N_{e,u}.
\end{equation*}
Moreover the sequence $\{p_N\}$
satisfies a large deviations principle with rate function
$A\colon (0,+\infty)\times \bb R^d\to [0,+\infty]$ given by
\begin{equation}
  \label{cra}
  A(e,u) = \sup_{\bs \gamma}\big\{ \bs\gamma \cdot (e,u) -
  \log m\big( \ee^{\bs \gamma\cdot\bs\zeta}\big) \big\}.
\end{equation}
This follows from the multidimensional Cram\'er's theorem
in \cite[Thm.~2.3.6]{DZ}.
Indeed,
in the terminology of convex analysis used in \cite{DZ},
the function $\bs \gamma \mapsto \log m\big(
\ee^{\bs \gamma \cdot \bs \zeta}\big)$
is \emph{steep}. Namely $|\nabla \log m\big(
\ee^{\bs \gamma \cdot \bs \zeta}\big)|$ diverges when
$\gamma_0 \to \gamma_0^*$. This follows from
item (ii) in Assumption \ref{ass:2}.

In view of Proposition~\ref{t:ldtp}
and the following remark,
Sanov's theorem for i.i.d.\ random variables
can be deduced from
Theorem~\ref{sld}.

\begin{remark}
  \label{remark1}
  We have
  \begin{equation*}
    \Ent(\pi|m) = \inf_{(e,u)} \big\{ A(e,u) + H_{e,u}(\pi)\big\}.
  \end{equation*}
  In fact, by a direct computation, the infimum is achieved for
  $(e,u) =\pi(\bs \zeta)$. 
\end{remark}

%%%%%%%%%%%%%%%%%%%%%%%%%%%%%%%%%%%%%%%%%%%%%%%%%%%%%%%%%%%%%%%%%%%%%%%%%%%%%%%%%%%%%%%%%%%%%%%%%%%%%%%%%%%%

\section{Large deviations for Kac model with microcanonical initial  data}
\label{sez:2}

\subsection*{The model}
Recall that $\Sigma^N = \big(\bb R^d\big)^N$.  We consider the Kac walk given
by the Markov process on the configuration space $\Sigma^N$,
whose generator acts on bounded continuous functions
$f\colon \Sigma^N\to \bb R$ as
\begin{equation*}
\mathcal L_N f(\bs v)=\frac 1 N \sum_{\{i, j\}} L_{i,j} f (\bs v),
\end{equation*}  
where the sum is carried over the unordered pairs
$\{i, j\}\subset \{1,.., N\}$, $i\neq j$, and
\begin{equation*}
L_{i,j} f(\bs v)= \int_{\bb S_{d-1}}\!\! \de \omega \,B(v_i- v_j, \omega)\big[f \big(T^{\omega}_{i,j} \bs v\big ) -f(\bs v)   \big].
\end{equation*}  
Here $\bb S_{d-1}$ is the sphere in $\bb R^d$ and  
\begin{equation}
  \label{rules}
  \big(T^{\omega}_{i,j} \bs v\big )_k = \begin{cases}
    v_i + (\omega \cdot (v_j-v_i))\omega  & \textrm{if } k=i\\
    v_j - (\omega \cdot (v_j-v_i))\omega  & \textrm{if } k=j\\
    v_k & \textrm{otherwise},
    \end{cases}
\end{equation}  
and the collision kernel $B$ is given by
\begin{equation}\label{eq:B}
  % B(w, \omega)=[w\cdot \omega]_- .
  B(v-v_*, \omega)=\frac 12 |(v-v_*)\cdot \omega|.
\end{equation}
Observe that the dynamics preserves energy and momentum, i.e. can be
restricted to the set $\Sigma^N_{e,u}$ as defined in \eqref{sig}. We
denote by $(\bs v(t))_{t\geq 0}$ the Markov process generated by
$\mathcal L_N$.

Fix hereafter $T>0$. Given a probability $\nu$ on $\Sigma^N$ we denote
by $\bb P_\nu^N$ the law of this process on the time interval $[0,T]$.
Observe that $\bb P_\nu^N$ is a probability on the Skorokhod
space $D([0,T];\Sigma^N)$.
As usual if $\nu=\delta_{\bs v}$ for some $\bs v \in \Sigma^N$, the
corresponding law is simply denoted by $\bb P_{\bs v}^N$.

\subsection*{Empirical observables}

Recall that $\ms P(\bb R^d)$ is the set of probability measures $\pi$ on
$\bb R^d$ equipped with the weak topology and the
corresponding Borel $\sigma$-algebra.
Let $D\big([0,T]; \ms P(\bb R^d)\big)$ the set of
$\ms P(\bb R^d)$-valued c{\'a}dl{\'a}g paths endowed with the
Skorokhod topology and the corresponding Borel $\sigma$-algebra.
Recalling the empirical measure $\pi^N$ defined in
\eqref{1}, with
a slight abuse of notation we denote also by $\pi^N$ the map from
$D\big([0,T]; \Sigma^N \big)$ to $D\big([0,T]; \ms P(\bb R^d)\big)$
defined by $\pi^N_t(\bs v)\coloneqq \pi^N(\bs v(t))$, $t\in [0,T]$.

%{\color{red}QUI CAMBIARE TOPOLOGIA?}
We denote by $\ms M$ the  subset of the finite measures $Q$  on
$[0,T]\times \bb R^{2d}\times \bb R^{2d}$  that satisfy
$Q(\de t; \de v,\de v_*, \de v',\de v_*')=Q(\de t; \de v_*,\de v, \de v',\de v_*') =Q(\de t; \de v,\de v_*, \de v'_*,\de v')$.
We consider $\ms M$ endowed with the weak topology and
the corresponding Borel $\sigma-$algebra. By definition, the weak topology is 
the weakest topology such that the map
$Q \mapsto Q(F)$ is continuous for each $F$ in
$C_{\mathrm b}([0,T]\times \bb R^{2d}\times \bb R^{2d})$.

% In order to fix the topology on $\ms M$, let
% $C_0 = C_0([0,T]\times \bb R^{2d}\times \bb R^{2d})$ be the closure of
% continuous compactely supported functions with respect to the uniform
% norm.  By Riesz representation theorem, its dual $C_0^*$ are the
% signed Borel measures on $[0,T]\times \bb R^{2d}\times \bb R^{2d}$.

% We endow $C_0^*$ with the bounded weak* topology, 
% namely a set is open if and only if its interesection with norm-bounded sets is open
% in the relative weak* topology.  The set $\ms M\subset C_0^*$ is then endowed with the relative
% bounded weak* topology and associated Borel $\sigma-$algebra. Then $\ms M$ is a completely regular
% topological
% space and, by the Banach-Alaoglu theorem, norm-bounded subsets of $\ms M$ are relative compact.

The \emph{empirical flow} is the map $Q^N\colon D\big([0,T]; \Sigma^N \big) \to \ms M$
defined by
\begin{equation}
  \label{2}
  Q^N(\bs v) (F) \coloneqq \frac 1N
  \sum_{\{i,j\}} \sum_{k\ge 1} F\big(\tau^{i,j}_k;
  v_i(\tau^{i,j}_k-),v_j({\tau^{i,j}_k}-),
  v_i(\tau^{i,j}_k),v_j(\tau^{i,j}_k)\big) 
  \quad 
\end{equation}
where $F\colon [0,T]\times \bb R^{2d}\times \bb R^{2d}\to \bb R$
is continuous, bounded, and satisfies $F(t; v, v_*, v',v_*')$ $=F(t;  v_*, v, v',v_*') = F(t;  v, v_*, v'_*,v')$,
while $(\tau^{i,j}_k)_{k\ge 1}$ are the
jump times of the pair $(v_i,v_j)$. Here, $v_i(t-) = \lim_{s\uparrow t} v_i(s)$.
In view of the conservation of the energy and momentum, the
measure $Q^N(\de t;\cdot)$ is supported on 
$\ms E \coloneqq \{
\bs \zeta(v)+\bs \zeta(v_*)=\bs \zeta(v')+\bs \zeta(v_*)
\}\subset \bb R^{2d}\times \bb R^{2d}$.

Let $\ms S$ be the  subset of
$D\big([0,T]; \ms P(\bb R^d)\big)\times \ms M$ given by elements
$(\pi,Q)$ that satisfies the balance equation
\begin{equation}
  \label{bal}
  \begin{split}
  &\pi_T(\phi_T)-\pi_0(\phi_0)-\int_0^T\! \de t\, 
  \pi_t(\partial_t \phi_t)\\
  &\qquad 
  +\int Q(\de t;\de v,\de v_*,\de v',\de v_*')\big[ \phi_t(v)+\phi_t(v_*)
  -\phi_t(v')-\phi_t(v_*')\big] =0
  \end{split}
\end{equation}
for each $\phi\in C_{\rm{b}}([0,T]\times \bb R^d)$
continuously differentiable in $t$, with bounded derivative.
For each $\bs v \in \Sigma^N$, with $\bb P^N_{\bs v}$ probability
one, the pair $(\pi^N,Q^N)$ belongs to $\ms S$.

\subsection*{The rate function}
Given $(e,u)\in Z$, recall
$C_{e,u} \coloneqq
\{\mu\in \ms P(\bb \R^d): \mu(\zeta_0) \le e, \, \mu(\zeta) = u\}$,
%that is a compact subset of
%$\ms P(\bb R^d)$.
and set
\begin{equation}
  \label{ceut}
  \ms C_{e,u} \coloneqq \{ \pi \in C([0,T],\ms P(\bb R^d):
  \pi_t \in C_{e,u}, t\in [0,T]\},
\end{equation}
that is a closed subset of $C([0,T],\ms P(\bb R^d)$.

For notation convenience, let $r(v,v^*,\cdot)$ be the measure on $\bb R^{2d}$
supported on
$\{\bs \zeta(v)+\bs \zeta(v_*)=\bs \zeta(v')+\bs \zeta(v_*')\}$
%$\{v+v_*=v'+v_*', \,|v|^2+|v_*|^2=|v'|^2+|v_*'|^2 \}$
such that
$$r(v,v_*,\de v', \de v_*') = \de \omega \, B(v-v_*,\omega),$$
%\id_{\{ \omega \cdot (v-v_*) \le  0\}},$$
%QUI PENSARE B CON IL MODULO E LEVARE INDICATRICE
where $v'$ and $v_*'$ are related to $\omega$ by the collision rules, as in  \eqref{rules}.
For $\pi\in D\big([0,T]; \ms P(\bb R^d)\big)$ let $Q^\pi$ be
the measure defined by
\begin{equation}
  \label{4}
  Q^\pi(\de t;\de v,\de v_*,\de v',\de v_*') \coloneqq  \frac 1 2 \de t \, \pi_t(\de v) \pi_t(\de v_*)
  \, r(v,v_*;\de v',\de v_*')
\end{equation}
and observe that $Q^\pi(\de t,\cdot)$ is supported on $\ms E$.

\begin{definition}
  \label{def:sac}
  Let $\ms S^\mathrm{ac}_{e,u}$ be the subset of $\ms S$ given by the elements
  $(\pi,Q)$ that satisfy the following conditions:
  \begin{itemize}
  \item [(i)] $\pi\in \ms C_{e,u}$; %%\big([0,T];\ms P_0(\bb R^d)\big)$;
  \item [(ii)] $Q\ll Q^\pi$.
  \end{itemize}
\end{definition}
Observe that if $(\pi,Q)\in \ms S^\mathrm{ac}_{e,u}$ then $Q^\pi$ is a finite measure.
% Moreover, by choosing positive functions $\phi$ not depending on $t$ in the balance
% equation \eqref{bal} and neglecting the loss term we obtain
% $$
% \pi_t(\phi)\leq \pi_0(\phi) +2\int_0^t\int Q(\de s, \de v, \de v_*, \de v', \de v'_*) \phi(v'). 
% $$
% Since $Q\ll Q^\pi$ and, by Assumption \ref{ass:1}, item (iii), the  marginal on $v'$ of $Q^\pi$ is absolutely continuous with respect to the
% Lebesgue measure, 
% we deduce that
% $\pi_0\ll \de v$ implies $\pi_t\ll \de v$, for any $t\geq 0$. As a
% consequence, also $Q$ is absolutely continuous with respect to the
% Lebesgue measure on $[0,T]\times \ms V$.
The dynamical rate function $J_{e,u}\colon \ms S \to [0,+\infty]$ is defined by
\begin{equation}
  \label{5}
  J_{e,u}(\pi,Q)\coloneqq 
  \begin{cases}
    {\displaystyle 
    \int  \de Q^\pi \Big[ 
    \, \frac{\de Q\phantom{^\pi}}{\de Q^\pi} \log \frac{\de Q\phantom{^\pi}}{\de Q^\pi} -
    \Big( \frac{\de Q\phantom{^\pi}}{\de Q^\pi}  -1\Big)\Big]  } & \textrm{if } (\pi,Q)\in \ms
    S^\mathrm{ac}_{e,u}\\ \\
    + \infty& \textrm{otherwise } 
  \end{cases}
\end{equation}
Recalling $H_{e,u}$ has been defined in \eqref{Heu}, the microcanonical
rate function is
\begin{equation}
  \label{I}
  I_{e,u}(\pi, Q) \coloneqq H_{e,u}(\pi_0) + J_{e,u} (\pi, Q).
\end{equation}
Let also $\hat {\ms S}$ be the subset of $\ms
S$ given by the pair $(\pi,Q)$ such that
\begin{equation*}
  \int_{[0,T]\times \bb R^{4d}} \!\de Q \, \big[\zeta_0(v) + \zeta_0(v_*) + \zeta_0(v') + \zeta_0(v'_*) \big] < +\infty.
\end{equation*}  
Observe that for
$(\pi, Q)\in \hat {\ms S}$ the balance equation \ref{bal} holds for
$\phi=\zeta_0$, therefore $\pi_t(\zeta_0)=\pi_0(\zeta_0)$ for every $t\in [0,T]$.

\begin{theorem}
  \label{upperbound}
  Assume $m$ satisfies condition {\rm (i)--(iii)} in Assumption \ref{ass:2}, 
  fix $(e,u) \in Z$, a sequence $(e_N,u_N) \to (e,u)$, and let
  $\nu^N_{e_N,u_N}$ be the microcanonical probabilities as in Section
  \ref{sez:1}.  The family
  $\bb P^N_{\nu^N_{e_N,u_N}} \circ (\pi^N,Q^N)^{-1}$ satisfies a large
  deviation upper bound with good rate function
  $I_{e,u}: \ms S \to [0,+\infty]$, namely $I_{e,u}$ has compact level
  sets and for each closed $C\subset \ms S$
  \begin{equation}\label{upeq}
    \varlimsup_{N\to +\infty} \frac 1N \log   \bb P^N_{\nu^N_{e_N,u_N}} \Big( (\pi^N,Q^N)\in C
    \Big) \le - \inf_{C} I_{e,u}.
  \end{equation}
  Moreover, if $m$ satisfies also
  condition {\rm (iv)} in Assumption \ref{ass:2}, then for
  each open  $O\subset \ms S$
  \begin{equation}\label{lbeq}
    \varliminf_{N\to +\infty} \frac 1N \log   \bb P^N_{\nu^N_{e_N,u_N}}
    \Big( (\pi^N,Q^N)\in O
    \Big) \ge - \inf_{O\cap \hat{\ms S}} I_{e,u}.
  \end{equation}
\end{theorem}

\section{Proof of the upper bound}
\label{sez:3}
The proof follows the same strategy as in \cite{BBBO} and in
\cite{He}.  For the reader convenience we here provide the details.
The upper bound is achieved by an established pattern in large
deviation theory. We first prove the exponential tightness, which
allows us to reduce to compacts. By an exponential tilting of the
measure, we prove an upper bound for open balls and finally we use a
mini-max argument to conclude.

The basic observation is the following. 
Given a bounded measurable function
$F\colon [0,T]
\times \bb R^{4d}\to \bb R$ such that
$F(t; v,v_*, v',v_*')=F(t; v_*,v, v',v_*')=F(t; v,v_*, v'_*,v')$,
set
\begin{equation}
  \label{def:lambdaF}
  \lambda^F(t;v,v_*)=\int r(v,v_*; \de v',\de v'_*)\ee^{F(t;v,v_*;v',v'_*)}.
\end{equation}
If $F=0$ we drop it from the notation. 
Denoting by $Q^N_{[0,t]}$ the restriction of the measure $Q^N$ on
$[0,t]$, and using that $\lambda(v,v) = \lambda^F(t,v,v) = 0$
%and setting
% $\vartheta(v)=\lambda(v,v)$, $\vartheta^F(t,v)=\lambda^F(t;v,v)$, $v\in\bb R^d$,
the
process
\begin{equation}\label{mart1}\begin{split}
      \bb M_t^F = & \exp\Big\{N\Big( Q_{[0,t]}^N(F)-\frac 1 2 \int_0^t \!\de s\,
      %\Big[
    \pi_s^N\otimes \pi^N_s\big(\lambda^F -\lambda  \big)
    %\\ & \int_0^t \!\de s\,
    %+\frac 1 N 
    %\pi^N_s\big(\vartheta^F-\vartheta\big)
    %\Big]
    \Big)
    \Big\}
  \end{split}\end{equation}
is a $\bb P^N_{\bs v}$ positive martingale for each $\bs v\in \Sigma^N$,
see e.g. \cite[App.~1, Prop.~2.6]{KL}.

For any $\delta>0$ we also define the compact set 
$C^\delta_{e,u} \coloneqq
\{\mu\in \ms P(\bb \R^d):
\mu(\zeta_0)\le e + \delta,\, |\mu(\zeta) - u|\le \delta \}$.
For $\delta>0$, by the conservation of the energy and the momentum
\begin{equation}
  \label{picdelta}
  \bb P^N_{\nu^N_{e_N,u_N}}( \pi_t^N \in C^\delta_{e,u}, t \in [0,T]) = 1,
\end{equation}
definitely in $N$.
By Ascoli-Arzel\`a and Prohorov theorems,
the exponential tightness follows from the next two lemmata.

\begin{lemma}\label{lemma2'}
  % There exist $c>0$ and $\ell_0>0$ such that for any $\ell > \ell_0$ and  $N >0$
  Set
  $$\bar F(v,v_*,v',v'_*) = \log ( 1 + \zeta_0(v) +
  \zeta_0(v_*) +  \zeta_0(v') +  \zeta_0(v'_*)).
  $$
  Then
  \begin{equation}
    \lim_{\ell\to+\infty}
    \varlimsup_{N\to+\infty}\frac 1 N \log{
      \bb P_{\nu^N_{e_N,u_N}}^N \Big ( Q^N(\bar F ) >  \ell \Big)}
    % \leq e^{-c N \ell}.
    =-\infty.
  \end{equation}
\end{lemma}  

%--------------------------------------------------------------------------------------------

\begin{lemma}\label{lemma3}
  For each $\epsilon >0$ and $\phi\in C_{\rm b}(\bb R^d)$
  \begin{equation}
    \lim_{\eta\downarrow 0}
    \varlimsup_{N\to+\infty}\frac 1 N \log{ \bb P^N_{\nu^N_{e_N,u_N}}\Big(\sup_{t, s \,\in [0,T]\;:|t-s|<\eta} |\pi^N_t(\phi)-\pi^N_s(\phi)|> \epsilon  \Big)}
    =-\infty.
  \end{equation}
\end{lemma}

%--------------------------------------------------------------------------------------------------------------------
% Lemmata \ref{lemma1} and  

%----------------lemma 2-------------------------

\begin{proof}[Proof of Lemma \ref{lemma2'}]
  Observe that by
  the conservation of the energy, %for each $\gamma \in [0,1/2]$,
  there exists a constant 
  $c >0$, depending on $e$, such that for any $N$ the bound  
  $Q^{\pi^N}(\ee^{\frac 12 \bar F}) \le c$ holds
  with $\bb P^N_{\nu^N_{e_N,u_N}}$ probability one.

  Let $M_T$ be the exponential martingale in \eqref{mart1} with  
  $F=\frac 12 \bar F$.
  Then, 
  for each $\ell >0$,  
  \begin{equation*}\begin{split}
      \bb P_{\nu^N_{e_N,u_N}}^N \Big ( Q^N(\bar F) > \ell  \Big)  %&
      =
      \bb E_{\nu^N_{e_N,u_N}}^N
      \Big(
      \bb M_T \,
      \big(\bb M_T\big)^{-1} \id_{Q^N(\bar F) > \ell}  \Big) %\\ &
      \leq \exp\{- N \ell/2 + cN\}.
  \end{split}\end{equation*}  
\end{proof}
%-------------------------------lemma 3---------------------------------------------------

\begin{proof}[Proof of Lemma \ref{lemma3}]
  In view of the balance equation \eqref{bal} 
  it is enough to show that 
  there exists a function $c\colon (0,1)\to \bb R_+$ with $c(\eta)\uparrow +\infty$ as $\eta\downarrow 0$ such that, for any $\epsilon >0$
  \begin{equation*}
    \bb P^N_{\nu^N_{e_N,u_N}}\Big(\sup_{t\in [0,T-\eta]} Q^N_{[t, t+\eta]}(1)> \epsilon
    \Big)\leq \ee^{-Nc(\eta)}.
  \end{equation*}
  By a straightforward inclusion of events, the previous bound follows from
  \begin{equation*}
    \frac 1 \eta \sup_{t\in[0,T-\eta]} \bb P^N_{\nu^N_{e_N,u_N}}\Big(Q^N_{[t, t+\eta]}(1)> \epsilon
        \Big)\leq \ee^{-Nc(\eta)}.\end{equation*}  
      Consider the super-martingale  \eqref{mart1} with $F=\gamma\,\id_{[t, t+\eta]}$, $\gamma >0$.
      Using the same argument of the previous lemma we deduce
\begin{equation*}
  \bb P^N_{\nu^N_{e_N,u_N}}\Big(Q^N_{[t, t+\eta]}(1)> \epsilon
  \Big)\leq
  \exp\Big\{-N\big[\gamma \epsilon - \eta\,(\ee^\gamma -1) C(1+e)  \big]  \Big\}.
\end{equation*}   
  The proof is concluded by choosing $\gamma= \log(1/\eta)$.
\end{proof}

\subsection*{Upper bound on compacts}
%Given a  bounded continuous function $\phi$ on $\R^d$
%For each bounded $\phi\in C(\bb R^d)$, let
%$\de m_{e_N,u_N}^\phi \coloneqq  \de m_{e_N,u_N} e^\phi
%/ m_{e_N,u_N}(e^\phi)$, with  $\de m_{e_N,u_N}$ is defined in \eqref{meu}.
%$\de m^\phi \coloneqq  \de m e^\phi/ m(e^\phi)$.% with  $\de m_{e_N,u_N}$ is defined in \eqref{meu}.
Recalling the set $C^\delta_{e,u}$ defined above \eqref{picdelta}, let $\mc C^\delta_{e,u}$ be the closed subset of $C([0,T]; \ms P(\bb R^d))$ defined as
\begin{equation}\label{C_stort}
\mc C^\delta_{e,u}\coloneqq \bigcap_{t\in [0,T]}\{\pi \colon \pi_t\in C^\delta_{e,u} \}.
\end{equation}
By Urysohn's lemma, for each $\eta >0$ there exists 
$\psi^{\delta, \eta}_{e,u}\colon C([0,T]; \ms P(\bb R^d))\to [0,1]$
continuous such that
\begin{equation*}
  \psi^{\delta,\eta}_{e,u}(\pi) =\begin{cases}
  0 \textrm { if}\quad \pi \in \mc C ^{\delta}_{e,u}\\
  1 \textrm { if} \quad  \mathrm{dist}(\pi, C^\delta_{e,u})\geq \eta,
  \end{cases}
\end{equation*}  
where $\mathrm{dist}$ is the uniform distance. Moreover, for
$\pi\in D([0,T], \ms P(\bb R^d))$, we extend it to a function defined on
$\bb R$ by setting $\pi_t=\pi_0$ if $t<0$, $\pi_t=\pi_T$ if $t>T$.
Let $\imath_\ve$ be the a smooth approximation of the $\delta$
function, and denote by $\imath_\ve * \pi$ the time convolution of
$\pi$.

\begin{lemma}
  \label{lemma:up-pf}
  Fix a measurable subset $B\subset \ms S$.
  For any $(\phi,F)\in C_{\mathrm b}(\R^d)\times C_{\mathrm b}(\R^+\times {(\R^{d})}^2\times {(\R^{d})}^2)$
  such that  and $F(t;v,v_*,v',v'_*)=F(t;v_*,v,v',v'_*)=F(t;v,v_*,v'_*,v')$,  and any $\delta, \eta, \ve, \alpha >0$,
  \begin{equation}
    \label{eq:up-pf}
    \varlimsup_{N\to \infty} \frac 1N \log\bb P^N_{\nu^N}
    \Big(
    (\pi^N,Q^N)\in B \Big)
    \le - \inf_{(\pi,Q)\in B} \big\{I_{\phi,F}(\pi,Q) +\alpha\psi^{\delta,\eta}_{e,u}\big(\imath_\ve * \pi\big)
    \big\},
  \end{equation}
  where
  \begin{equation}
    \label{eq:Ipf}
    I_{\phi,F}(\pi,Q)\coloneqq  \pi_0(\phi) - \Lambda_{e,u}(\phi)+ 
    Q(F)-\frac 1 2 \int_0^T \de t \, \pi_t\otimes \pi_t(\lambda^F-\lambda).
    \end{equation}
\end{lemma}
  \begin{proof}
    Let ${\tilde \nu}^N_{e_N,u_N}$ be the probability on $\Sigma^N$ defined by
    $$
\de{\tilde \nu}^N_{e_N,u_N}=\de {\nu}^N_{e_N,u_N}\exp\big\{N \pi^N(\phi)-\log \nu^N_{e_N, u_N}\big(\ee^{N\pi^N(\phi)}\big)\big\}.
    $$
    Recalling \eqref{picdelta} and  the definition of the martingale $\bb M^F_t$ in
    \eqref{mart1}, we write
    \begin{equation*}\begin{split}
    \bb P^N_{\nu^N_{e_N,u_N}}
    \Big(
    (\pi^N,Q^N)\in B \Big)=
     \int \nu^N_{e_N,u_N} (\de {\bf v}) \bb E^N_{\bf v}\left( \ee^{-N\alpha \psi^{\delta,\eta}_{e,u}(\imath_\ve * \pi^N)}
\id_{B}(\pi^N,Q^N)\right)\\
=\int \de {\tilde\nu}^N_{e_N,u_N}\,\frac{\de \nu^N_{e_N,u_N}}{ \de {\tilde\nu}^N_{e_N,u_N}} \, \bb E^N_{\bf v}\left(\ee^{-N\alpha \psi^{\delta,\eta}_{e,u}(\imath_\ve * \pi^N)}\bb M^F_T \big(\bb M^F_T   \big)^{-1} \id_{B}(\pi^N,Q^N)\right)
    \end{split}\end{equation*}  
    %
    % Recalling that  $\vartheta(v)\coloneqq \lambda(v,v)$,
    We get
  \begin{equation*}
  \begin{split}
    &\bb P^N_{\nu^N_{e_N,u_N}}
    \Big(
    (\pi^N,Q^N)\in B \Big)
    \\
    &\le  \sup_{(\pi,Q)\in B}
    \exp \big\{-N \big[\pi_0(\phi)-\frac 1 N\log \nu^N_{e_N, u_N}\big(\ee^{N\pi^N(\phi)}\big)+\alpha\psi^{\delta,\eta}_{e,u}\big(\imath_\ve * \pi\big)\\
      &
\qquad \qquad \qquad +Q(F)-\frac 1 2 \int_0^T\! \de t \,\pi_t\otimes \pi_t(\lambda^F-\lambda)
\big]\big\},\\
%& \qquad \times
%\bb E_{{\tilde \nu}^N_{e_N, u_N}}^N (e^{\frac 1 2 \int_0^T \de t\, \pi^N_t (\vartheta^F - \vartheta)} \bb M_T^F).
  \end{split}
\end{equation*}
where we used that $\bb E_{{\tilde \nu}^N_{e_N, u_N}}^N ( \bb M_T^F)= 1$.
The statement follows from Lemma \ref{lemma1}. 
\end{proof}

%Recall that $H(\cdot|m)$ denotes the relative entropy and let $J$ be the functional defined in \eqref{5}.
\begin{lemma}[Variational characterization of the dynamical rate functional]
  \label{p:vr} 
  For any pair
  $(\pi,Q)\in\ms S$ such that $\pi\in C([0,T];\ms P(\bb R^d))$ 
  \begin{equation}
      J_{e,u}(\pi,Q) = \sup_{F, \alpha, \delta, \eta, \ve} \Big\{Q(F)-\frac 1 2 \int_0^T \de t \, \pi_t\otimes
      \pi_t(\lambda^F-\lambda) + \alpha\psi^{\delta,\eta}_{e,u}\big(\imath_\ve * \pi\big)\Big\},
   \end{equation}
where  the supremum is carried out  over all continuous
and bounded $F\colon [0,T]\times (\bb R^d)^2\times  (\bb R^d)^2 \to\bb R$
such that $F(t;v,v_*,v',v'_*)=F(t;v_*,v,v',v'_*)=F(t;v,v_*,v'_*,v')$, and $\alpha, \delta, \eta, \ve >0$.
\end{lemma}

\begin{proof}
  By monotonicity
  \begin{equation*}
  \sup_{\alpha, \delta, \eta, \ve} \alpha\psi^{\delta,\eta}_{e,u}\big(\imath_\ve * \pi\big) =\sup_{\ve} \lim_{\alpha\uparrow +\infty}\lim_{\delta\downarrow 0}\lim_{\eta\downarrow 0}
  \alpha\psi^{\delta,\eta}_{e,u}\big(\imath_\ve * \pi\big)=\begin{cases}
  0 \textrm{ if } \pi\in\mc C_{e,u}\\
  +\infty \textrm{ otherwise},
  \end{cases}
  \end{equation*}
  where we have used that if $(\imath_\ve * \pi)_t\in C_{e,u}$, for any $t\in [0,T]$ and $\ve>0$, then $\pi_t \in C_{e,u}$ for any $t\in [0,T]$.
  To complete the proof, it remains to show that for $\pi\in \mc C_{e,u}$
  \begin{equation}
    \label{true-var}
J_{e,u}(\pi, Q)=\sup_{F} \Big\{Q(F)-\frac 1 2 \int_0^T \de t \, \pi_t\otimes
      \pi_t(\lambda^F-\lambda)\Big\}.
  \end{equation}
 Recall the definition of $Q^\pi$
  in \eqref{4} and observe that
  \begin{equation*}
   \frac 1 2 \int_0^T \de t \, \pi_t\otimes \pi_t(\lambda^F-\lambda)
    = Q^\pi \left( \ee^F - 1 \right).
  \end{equation*}
  This implies that if $\sup_F \Big[Q(F)-\frac 1 2 \int_0^T \de t \, \pi_t\otimes
  \pi_t(\lambda^F-\lambda)\Big]$ is finite, then
  $Q$ is absolutely continuous with respect to
  $Q^\pi$.
 The proof is now completed by a direct computation.
\end{proof}

\begin{proof}[Proof of Theorem \ref{upperbound}, upper bound]
  In view of \eqref{picdelta}, Lemma \ref{lemma2'} and Lemma \ref{lemma3} imply the exponential tightness of the family $\{\bb P^N_{\nu^N_{e_N, u_N}}\circ (\pi^N, Q^N)^{-1}\}$.
 Moreover,  Lemma \ref{lemma3} implies that if the large deviation upper bound rate function is finite then
$\pi \in C([0,T], \ms P(\bb R^d))$.
Therefore it is enough to show the statement
  for compacts. In view of Lemma \ref{lemma:up-pf} and the 
  mini-max argument in
  \cite[App.2, Lemma~3.2]{KL}, the statement follows from Lemma
  \ref{lemma2} and Lemma \ref{p:vr}.
\end{proof}

\section{Proof of the lower bound}
\label{sez:4}
In this section we adapt the strategy in \cite{BBBO} to the Kac model,
where the kernel $B$ is not strictly positive.
We shall first prove the lower bound for open neighborhoods of ``nice''
$(\pi, Q)$, and then use a density argument.
As in \cite{BBBO} and \cite{He} we will restrict to $Q$ with bounded
second moment, but we will not require, as in \cite{He}, that $B\ge c > 0$.

\subsection*{Perturbed Kac walks} 
We start by the following law of large numbers for a class of
perturbed Kac's walks.
Consider perturbed time dependent collision kernels
$\tilde B$ that are continuous and satisfy 
\begin{equation}\label{v'}
\sup_{t,v,v_*} 
\tilde\lambda_t(v, v_*)=
\sup_{t,v,v_*} 
\int \tilde B_t (v,v_*,\omega) \de \omega \leq C,
\end{equation}  
for some $C<+\infty$.
Fix $(e,u) \in Z$, a sequence
$(e_N,u_N) \to (e,u)$, and
let $\nu^N_{e_N,u_N}$ be the family of probabilities on $\Sigma^N$ as in
Section \ref{sez:1}, and denote by 
$\tilde{\bb P}_\nu^N$
the law of the perturbed Kac walk with initial datum $\nu^N_{e_N,u_N}$.

\begin{lemma}\label{l:perk}
  As $N\to +\infty$, the pair $(\pi^N,Q^N)$ converges,
  in $\tilde{\bb P}_{\nu^N}^N$ probability, to
  $(f \de v\,, q\, \de t \de v\de v_* \de \omega)$, where
  $q_t(v,v_*,\omega) = \frac 12 f_t(v) f_t(v_*) \tilde B_t(v,v_*,\omega)$ and
  $f\in C\big([0,T]; L^1(\bb R^d)\big)$ is the
  unique  solution   to the
  perturbed Kac's equation
  \begin{equation}
    \label{perk}
    \begin{cases}
      \vspace{3pt}
      {\displaystyle \partial_t f_t(v) =
        \iint \!\de v_* \de \omega \,
        \big[\tilde B_t(v',v'_*,\omega) f_t(v') f_t(v_*')  -
        \tilde B_t(v,v_*,\omega)  f_t(v) f_t(v_*) \big]}, \\
      f_0(\cdot) = \frac{\de m_{\mathrlap{e,u}}}{\de v}. 
    \end{cases} 
  \end{equation}  
  Here we understand that \eqref{perk} holds
  by integrating against  continuous, bounded test functions
  which are continuous differentiable in time.
\end{lemma}  
The proof follow from the fact the
large deviation upper bound
holds also for the perturbed Kac's walk, and
the uniqueness of the solution due \eqref{v'},
see proof of Lemma 4.1 in \cite{BBBO} for the details.

The following specifies the collection of ``nice'' $(\pi, Q)$.
Recall
${\mc S}^{\textrm{ac}}_{e,u}$ in Definition \ref{def:sac}.
\begin{definition}\label{def:B}
Let  $\tilde{\ms S}_{e,u}$ be the collection of elements
$(\pi, Q)\in \ms S^{\mathrm{ac}}_{e,u}$ whose densities $(f, q)$ are
continuous and such
  that
\begin{equation}\label{def:B1}
\sup_{t,v,v_*,\omega} 
\frac {q_t(v,\,v_*,\omega)}
  {f_t(v) f_t(v_*)}< +\infty,
\end{equation}
%% %for some constant $C>0$,
and
\begin{equation}\label{def:B2}
  \sup_{t,v,v_*,\omega}
    \frac {q_t(v,\,v_*, \omega)}
  {f_t(v) f_t(v_*)B(v,v_*,\omega)}< +\infty.
\end{equation}
\end{definition}

Given $(\pi, Q)\in \tilde{\ms S}_{e,u}$, denote by $\tilde B_t$ the time
dependent perturbed kernel defined by 
\begin{equation}\label{def:Bt}
\tilde  B_t(v, v_*, \omega) = 2\frac {q_t(v,\,v_*,\omega)}{f_t(v) f_t(v_*)},
\end{equation}
that  meets  \eqref{v'}.

The next statement provides
the large deviation lower bound for  neighborhood of elements in $\tilde{\ms S}_{e,u}$.

\begin{proposition}\label{l:lb}
  Let  $(\pi, Q)\in \tilde{\ms S}_{e,u}$. Assume that $\pi_0$ satisfies items
  (iii) in Assumption \ref{ass:2}, and suppose $\pi_0(\de v)=\ee^\phi
  m(\de v)/m(\ee^\phi)$ for some $\phi$ bounded and continuous.
  Fix a sequence
  $(e_N,u_N) \to (e,u)$, and denote by
  $\tilde \nu^N_{e_N,u_N}$ the regular version of the probability
  $\pi_0^{\otimes N}$ conditioned to
  $(\frac 1 N\sum_{i=1}^N \frac 1 2 { |v_i|^2}, \,\frac 1 N
  \sum_{i=1}^Nv_i)$ evaluated at  $(e_N,u_N)$. 
  Then 
  $$
  \varlimsup_{N\to\infty}
  \frac 1 N \Ent\Big(\tilde{\mathbb P}^N_{{\tilde \nu}^N_{e_N,u_N}}\vert
  \mathbb P^N_{\nu^N_{e_N,u_N}} \Big) =
  I_{e,u}(\pi, Q).
  $$
\end{proposition}

We premise the following Lemma.

\begin{lemma}
  \label{lemmanext}
  If $F\in C_{\mathrm b}([0,T]\times \bb R^{d}\times \bb R^d\times S_{d-1})$,
  %  \begin{equation*}
%  \esssup_{t} \sup_{v,v_*,\omega} |F_t(v,v_*,\omega)| <+\infty.    
%  \end{equation*}
  then 
  \begin{equation*}
    \limsup_{N\to\infty}
    \tilde {\bb E}^N_{{\tilde \nu}^N_{e_N,u_N}}\big( Q^N(F)^2 \big) < +\infty.
\end{equation*}
  
\end{lemma}
\begin{proof}
    Set
  \begin{equation*}
    \tilde{M}^{N}_t\coloneqq  Q^N_{[0,t]}(F)-\frac 1 {N^2}\sum_{\{i,j\}}
    \int_0^t \!\de s\,\int\! \de \omega\, \tilde B_s (v_i,v_j, \omega)
    F_s(v_i,v_j,\omega),
  \end{equation*}
  that it is a $\tilde {\bb P}_{{\tilde \nu}^N_{e_N,u_N}}$ martingale
  with predictable quadratic variation 
  \begin{equation*}
    \langle \tilde{M}^{N}  \rangle_t = \frac 1 {N^2}\sum_{\{i,j\}}
   \int_0^t \!\de s\,\int\! \de \omega\, \tilde B_s (v_i,v_j, \omega)
    F_s(v_i,v_j,\omega)^2.
   \end{equation*} 
   In view of \eqref{v'}, the random variable
   $\langle \tilde{M}^{N}  \rangle_T $ is uniformly bounded in $N$,
   which implies the statement.
\end{proof}

\begin{proof}[Proof of Proposition~\ref{l:lb}]
  By using Theorem \ref{sld}, it is enough to show that 
  \begin{equation}\label{HJ}
    \varlimsup_{N\to\infty}
  \frac 1 N \Ent\Big(\tilde{\mathbb P}^N_{{\tilde \nu}^N_{e_N,u_N}}\vert
  \mathbb P^N_{{\tilde \nu}^N_{e_N,u_N}} \Big) =
  J_{e,u}(\pi, Q).
  \end{equation}
  In view of the assumptions on $\tilde B$, the
  value at time $T$ of the 
  martingale
  defined in \eqref{mart1} with $F_t=\log(\tilde B/B)$ is
  the Radon-Nykodim
  derivative of $\tilde{\bb P} ^N_{{\tilde\nu}^N_{e_N,u_N}}$
  with respect to
  $\bb P ^N_{{\tilde\nu}^N_{e_N,u_N}}$.
  Since $\lambda_t^F=\tilde \lambda_t$,
  %and
  %$\vartheta^F_t(v)=\tilde\lambda_t(v,v)=: \tilde\vartheta_t$,
  \begin{equation*}
    \begin{split}
      &\frac 1 N \Ent \Big(\tilde{\mathbb P}^N_{{\tilde \nu}^N_{e_N,u_N}}\vert
      \mathbb P^N_{{\tilde \nu}^N_{e_N,u_N}} \Big) \\
      &=\tilde {\bb E}_{{\tilde \nu}^N_{e_N,u_N}}
      \Big( Q_{[0,T]}^N(F)-\frac 1 2 \int_0^T \!\de s\,
      %\Big[
    \pi_s^N\otimes \pi^N_s\big(\tilde \lambda_s -\lambda  \big)
    %\\ & \int_0^t \!\de s\,
    %+\frac 1 N 
    %\pi^N_s\big(\tilde \vartheta_s\big)\Big]
    \Big). 
  \end{split}\end{equation*}
Now observe that, by Lemma \ref{l:perk}, $(\pi^N,Q^N)$ converges to
$(\pi,Q)$ in $\tilde{\mathbb P}^N_{{\tilde \nu}^N_{e_N,u_N}}$
probability.
By definition of $\tilde S_{e,u}$, $F$ satisfies the assumption of
Lemma~\ref{lemmanext}, then the sequence $Q_{[0,T]}^N(F)$ is uniformly
integrable with respect to $\tilde{\mathbb P}^N_{{\tilde \nu}^N_{e_N,u_N}}$.
By \eqref{v'}, 
$\pi^N_s\otimes \pi^N_s(\tilde\lambda_s)$ converges to
$\pi_s\otimes\pi_s(\tilde\lambda_s)$ for almost all $s\in[0,T]$. 
Moreover, by conservation of energy, 
$\lambda$ is uniformly integrable with respect to 
$\de s\,\pi^N_s\otimes \pi^N_s$.
%Finally, again by \eqref{v'}, ${\esssup}_{s}\sup_v
%\tilde \vartheta_s(v) <+\infty$.
Therefore \eqref{HJ} follows.
\end{proof}

\subsection*{Approximating paths}

Recall that  the set $\hat{\ms S}$ has been defined above Theorem
\ref{upperbound}.
\begin{theorem}\label{the:approx}
  For each  $(\pi, Q)\in \hat{\ms S}$ such that  $I_{e,u}(\pi,Q) < +\infty$
  there exists a sequence $\{(\pi_n,Q_n)\}\subset
  \tilde{\ms S}_{e,u}\cap \hat{\ms S}$ satisfying
  $(\pi_n,Q_n)\to (\pi, Q)$ and $I_{e,u}(\pi_n,Q_n)\to I(\pi, Q)$.
\end{theorem}

\begin{proof}
  The proof is achieved by combining the following
  three steps
  and a standard diagonal argument.
  In particular, in Step 1 we construct
  positive regular approximating probability paths,
  in Step 2 we regularize in time, 
  in Step 3 we perform
  a truncation argument as in \cite{BBBO},
  adapted to the hard-sphere kernel.

  \smallskip\noindent\emph{Step 1. Velocity convolution.}
  Since $I_{e,u}(\pi,Q)<+\infty$ and $(\pi,Q)\in \hat {\ms S}$,
  $\pi_t(\zeta ) = u$, $\pi_t(\zeta_0) = \pi_0 (\zeta_0)
  = U + |u|^2/2  \in (0,e]$,
  where $U = \frac 12 \int \pi_t(\de v) |v-u|^2$ is the internal energy.

  Let $(f,q)$ be the densities of $(\pi, Q)$.
  Given $0<\delta<1$, let $g_\delta$ be the Gaussian kernel on
  $\bb R^d$ with variance $\delta$ and define
  \begin{equation}
    \label{def-fdelta}\begin{split}
      & f^\delta_t(v)=%\int\! \de s\, \chi_\delta(t-s)
      \alpha (g_\delta * f_t)(\alpha (v -u )+u)\\
      & q_t^\delta(v,v_*,\omega) =
      \alpha^2 (g_\delta\otimes g_\delta \otimes \opid * 
      q)(\alpha (v-u) + u , \alpha (v_*-u) +u , \omega)
    \end{split} 
  \end{equation}
  where $\opid$ is the identity function and 
  $\alpha = \alpha(\delta) > 0$ is chosen such that
  $\int \de v f^\delta_t(v) |v-u|^2/2= U$. Observe that
  for any $\alpha>0$, 
  $\int \de v f_t^\delta (v)  v = u $.

  Let $(\pi^\delta,Q^\delta)$
  be the 
  pair with densities $(f_t^\delta,q_t^\delta)$, which
  satisfies the balance equation.
  In order to prove the convergence of the rate function, we first
  observe
  that, 
  by item (ii) in
  Assumption \ref{ass:2}, we can write 
  $$
  \Ent(\pi^\delta_0| m_{e,u})=
  \int f^\delta_0\log f^\delta_0 + \int f^\delta_0 \log \frac {1} {m_{e,u}}.
  $$
  Since $\alpha(\delta) \to 1$ as $\delta \to 0$,
  by Jensen inequality and item
  (ii) in
  Assumption \ref{ass:2},
  \begin{equation*}
%    \label{ent-delta}
    \varlimsup_{\delta \to 0} \Ent( \pi^\delta_0| m_{e,u})
    \le \Ent(\pi_0|m_{e,u}).
  \end{equation*}
  By the choice of $\alpha$, $f^\delta_0$ has the same energy
  as $f_0$. Therefore 
  \begin{equation*}
   % \label{ent-delta}
    \varlimsup_{\delta \to 0} H_{e,u}( \pi^\delta_0)
    \le H_{e,u}(\pi_0).
  \end{equation*}
  
  We will conclude the proof showing that
  $\varlimsup J_{e,u} (f^\delta,q^\delta) \le J_{e,u} (f,q)$.
  We first observe that by
  a straightforward approximation argument we can 
  choose $F=\log 1/B$ in \eqref{true-var}, and deduce 
  \begin{equation}
    \label{eq:qlog1b}
  Q\big(\log \frac 1B\big) \le  J_{e,u} (f,q) + 
  \frac 12 \int_0^T \de t \int \de v \de v_* \de \omega
  f f_* B \big(\frac 1B - 1\big) < \infty.
  \end{equation}
  We prove in Appendix \ref{appendixa} that
  $Q^\delta(\log 1/B)$
  is bounded and 
  converges to $Q(\log 1/B)$ as $\delta \to 0$. Therefore
%  for $\delta$ small enough
%  Since $Q^\delta(\log 1/B)$ is bounded, 
  $$J_{e,u}(\pi^\delta, Q^\delta)
  = \int_0^T \de t \int \de v \de v_* \de \omega\,
  q^\delta \log \frac {2 q^\delta}{f^\delta f^\delta_*}
  + Q^\delta \big( \log \frac 1B \big) - Q^\delta(1) +
  Q^{\pi^\delta}(1).$$
  Since  the map $[ 0, +\infty)^2 \ni (a,b)\mapsto a\log(a/b)$
  is one-homogeneous and convex, by \eqref{def-fdelta} 
  and Jensen's inequality
  the first term on the r.h.s. is bounded by
  $Q(\log \frac {2q}{ff_*})$.
  Moreover, $Q^\delta(1) = Q(1)$, while,
  since $B= \frac 12 |(v-v_*)\cdot \omega|$, 
  $Q^{\pi^\delta} (1) = \frac 1\alpha Q^\pi(1)$.
  
  \smallskip\noindent\emph{Step 2. Time convolution.}
  Consider $(\pi,Q)\in \hat {\ms S}$ such that $I_{e,u}(\pi,Q) < +\infty$,
  and denote with $(f,q)$ their densities. Assume
  that $f$ and $q$ are smooth in the velocities, and $f>0$.
  Observe that approximating path constructed in Step 1 meets
  these requirements.

  Extend $[0,T] \ni t\mapsto (f_t,q_t)$ to a function defined on
  $(-\infty,T]$ by setting
  $(f_t,q_t) = (f_0,0)$ if $t<0$.
  Let $\imath_\ve$ be the a smooth approximation of the
  $\delta$ function, with support in $(-\ve,0)$,
  and denote by $(\pi^\ve,Q^\ve)$ the path
  with densities $(f^\ve,q^\ve) = \imath_\ve * (f,q)$; here we
  understand the convolution in time.
  The pair 
  $(\pi^\ve,Q^\ve)$ converges to $(\pi,Q)$ and satisfies the balance equation
  \eqref{bal}.
  Observe that $f^\eps_0 = f_0$ and,  since $(\pi,Q)\in \hat {\ms S}$,
  $\pi_t(\zeta) = \pi_0(\zeta)$ for any $t\in [0,T]$, 
  so that $\pi^\ve_t(\zeta) = \pi_0(\zeta)$
  for any $t\in [0,T]$,

  We claim that
  $\lim_{\ve \to 0} I_{e,u}(\pi^\ve,Q^\ve) = I_{e,u} (\pi,Q)$.
  To this hand, as $H_{e,u}(\pi^\ve_0) =  H_{e,u}(\pi_0)$,   
  by lower semi-continuity it is enough to show that
  $\varlimsup_{\ve \to 0} J_{e,u}(\pi^\ve,Q^\ve) \le J_{e,u} (\pi,Q)$.

  Let $g_1$ be the standard Gaussian density on $\bb R^d$.
  We observe that, by standard approximation argument,
%  since $(\pi,Q)\in \hat {\ms S}$ and $f>0$, 
  we can choose $F= \log g_1 / f$ in the variational formula
  \eqref{true-var}, and deduce that
  $\int q \log \frac 1f< + \infty$ is finite.
  Since $J_{e,u}(\pi,Q)$ is bounded, using \eqref{eq:qlog1b},
  we then deduce that $\int q \log q < +\infty$.

  By Jensen inequality
  $\int q^\ve  \log q^\ve  \le  \int q \log q < +\infty$.
  On the other hand, by convexity, the maps
  $q \mapsto \int q \log q$ is lower semi-continuous, 
  therefore we conclude that
  $$\lim_{\ve \to 0}   \int q^\ve  \log q^\ve  =  \int q \log q.$$

  We write
  \begin{equation*}
  \begin{split}
    J_{e,u} (\pi^\ve, Q^\ve) =
    &-\int q^\ve \log 2q^\ve +
    \int q^\ve \log \frac {2q^\ve}{f^{\ve}} +
    \int q^\ve \log \frac {2q^\ve}{f_*^{\ve}} \\
    &+
    \int q^\ve \big( \log \frac 1B -1\big)
    + \int f^\ve f^\ve_* B.
  \end{split}
  \end{equation*}
  As already stated, the first term on the right-hand-side converges.
  By Jensen inequality the second term is  bounded by
  $\int q \log (2q/f)$ 
  and the  third by $\int q \log (2q/f_*)$.
  Moreover,  the fourth does not depend on $\ve$.
  The convergence of the last term follows from the fact that,
  since the energy is uniformly bounded
  and $\pi\in C([0,T],\ms P(\bb R^d)$, the map
  $[0,T]^2\ni (s,s') \mapsto \int \de v \de v_* \de \omega
  f_s(v) f_s(v_*) B(v-v_*,\omega)$ is continuous.

  \smallskip\noindent\emph{Step 3. Truncation.}
  Consider $(\pi,Q)\in \hat {\ms S}$ with
  $I_{e,u}(\pi, Q)<+\infty$, with densities $(f,q)$.
  We denote
  by $q_t^{(i)}$, $i=1,\ldots ,4$ the marginal of $q_t$
  respectively on $v,v_*,v',v_*'$. Then
  $q_t^{(1)}= q_t^{(2)}$,
  $q_t^{(3)}= q_t^{(4)}$, and the balance equation is the weak version of
  the identity
  $$\partial_t f_t = 2\big( q_t^{(3)} - q_t^{(1)}\big).$$
  In the sequel we
  assume  $(f,q)$ smooth, %such that $I_{e,u}(\pi,Q)<+\infty$,
  $f$ strictly
  positive, 
  and $q_t^{(3)}\in L^2([0,T]\times \bb R^d)$.
  Observe that the approximating path defined by
  applying sequentially Step 1 and 2
  meets the above conditions. Indeed,
  the last condition above follows by  
  Young inequality for convolutions.

  Given $\ell>0$,
  let  $\chi^\ell(v,v_*,\omega)\in [0,1]$ be a continuous
  function  such that
  $$\chi^\ell (v,v_*,\omega) =
  \begin{cases}
    1 & \text{ if } |v|^2 + |v|^2_* < \ell \text { and }
    |(v-v_*)\cdot \omega| > 1/\ell \\
    0 & \text{ if } |v|^2 + |v|^2_* \ge (\ell+1) \text { or }
    |(v-v_*)\cdot \omega| \le 1/(\ell+1) 
  \end{cases}
  $$
  We define $(\tilde f^\ell, \tilde q^\ell)$ by
  \begin{equation}
    \label{tildeell}
    \begin{aligned}
      &\tilde q^\ell(v,v_*,\omega) = q(v,v_*,\omega) \chi^\ell(v,v_*,\omega)\\
    &\tilde f^\ell_t = f_0 +
    2
      \int_0^t \de s
      \left(\tilde q^{\ell,(3)}_s-
        \tilde q^{\ell,(1)}_s
      \right)
      + 2\int_0^T  \de s
      \left(q^{(3)}_s-
        \tilde q^{\ell,(3)}_s
      \right)
  \end{aligned}
\end{equation}
Observe that $\tilde q_t^\ell \le q_t$. Moreover
$\tilde f^\ell_t \ge  f_t $, since
\begin{equation}
  \label{variqtilde}
  \begin{aligned}
    &\int_0^t \de s
    \left(\tilde q^{\ell,(3)}_s-
      \tilde q^{\ell,(1)}_s
    \right)
  + \int_0^T  \de s
  \left(q^{(3)}_s-
    \tilde q^{\ell,(3)}_s
  \right)\\
  &=
  \int_0^t \de s
  \left( q^{(3)}_s-
    \tilde q^{\ell,(1)}_s\right)
    +
    \int_t^T  \de s
    \left(q^{(3)}_s-
      \tilde q^{\ell,(3)}_s
    \right).
    \end{aligned}
\end{equation}
Set
$$c_\ell^{-1} = 1 +  2\int_0^T  \de s \int \de v
\left(q^{(3)}_s-
  \tilde q^{\ell,(3)}_s
\right), %= 1 + Q(\Omega^c_\ell).
$$
and denote by $(e_\ell,u_\ell)$  the energy and momentum
of the probability $c_\ell {\tilde f}_t^\ell \de v$. Note that
$(e_\ell,u_\ell)$ does not depend on $t$ since $(\pi,Q)\in \hat{\ms S}$.
We define
$(f^\ell, q^\ell)$ by:
$$f^\ell(v) = \alpha  c_\ell \tilde f^\ell( \alpha (v-u)+u_\ell), \ \
\ q^\ell(v,v_*,\omega) = \alpha^2 c_\ell \tilde q^\ell
( \alpha (v-u)+u_\ell, \alpha (v-u)+u_\ell)$$
where $\alpha=\alpha_\ell>0$ is 
chosen such that
$\int \de v f_0^\ell (v) \bs \zeta(v)=\int \de v f_0 (v) \bs \zeta(v)$.
Observe that
the pair $(f^\ell, q^\ell)$ satisfies the balance
equation. 
As $\ell\to +\infty$, $c_\ell \to 1$, $u_\ell\to u$,
$\alpha_\ell\to 1$,
therefore $(f^\ell, q^\ell)$ converges to $(f,q)$.

We claim that
$$\varlimsup_{\ell \to +\infty} I_{e,u}(\pi^\ell,Q^\ell) \le
I_{e,u} (\pi,Q).$$
We start by proving that
\begin{equation}\label{dario3}
\varlimsup_{\ell \to +\infty}
H_{e,u}(\pi_0^\ell) \le H_{e,u}(\pi_0).
\end{equation}
Let $m^\ell$ be the probability measure satisfying
$$\int m^\ell (\de v) \varphi(v) =
\int m_{e,u} (\de v) \alpha \varphi(\alpha (v-u) + u_\ell),$$
for any $\varphi \in C_{\mathrm b}(\bb R^d)$, 
and let $\rho^\ell$ be its density.
By a change of variable 
\begin{equation}
  \label{entconvex}
  \Ent(\pi^\ell_0|m_{e,u}) = \Ent(c_\ell \tilde f^\ell_0 \de v| m^\ell).
\end{equation}
By \eqref{tildeell}, 
$$c_\ell \tilde f_0^\ell = c_\ell f_0 + (1-c_\ell) \bar h^\ell,$$
where
$h^\ell = 
2\int_0^T  \de s
\left(q^{(3)}_s-
  \tilde q^{\ell,(3)}_s
\right)
$ and $\bar h^\ell =  h^\ell / \int h^\ell$.
By convexity
$$\Ent(c_\ell \tilde f^\ell_0 \de v| m^\ell)
\le c_\ell \Ent (\pi_0| m^\ell) + (1-c_\ell) \Ent( \bar h^\ell\de v
|m^\ell).$$
Since $c_\ell\to 1$, $\alpha_\ell\to 1$, $u_\ell \to u$,
in view of item (iv) in Assumption \ref{ass:2},
by dominated convergence the first term on the right-hand-side
of \eqref{entconvex} converges to 
$\Ent(\pi_0|m)$.

We now show that the second term vanishes. 
%Since $c_\ell\to 1$, by the convexity of $H(\cdot | m)$ it is enough to show
%$$(1-c_\ell) \Ent(\bar h^\ell|m)\to 0.$$
Observe that
$$
(1-c_\ell) \Ent(\bar h^\ell\de v|m^\ell)
= c_\ell \int h^\ell \log h^\ell +(1-c_\ell)\log \frac{c_\ell}{1-c_\ell}
- c_\ell \int h^\ell \log \rho^\ell
$$
Since, by assumption on $q^{(3)}$, $h^\ell \in L^2$ and
it converges to zero pointwise,
the first term vanishes. The second term vanishes since $c_\ell \to 1$.
Finally, using item (iv)
of Assumption \ref{ass:2},
the last term vanishes
by dominated convergence.
Since $\pi_0^\ell(\bs \zeta) = \pi_0(\bs \zeta)$, \eqref{dario3} follows.

We conclude the proof by   showing that
$$\lim_{\ell \to +\infty} J_{e,u}(\pi^\ell,Q^\ell) = J_{e,u} (\pi,Q).
$$
By a change of variables, 
\begin{equation*}
  \begin{split}
    J_{e,u}(\pi^\ell,Q^\ell) =
    c_\ell \int  \tilde q^\ell \log \frac { 2 \tilde q^\ell }{
      c_\ell \tilde f^\ell \tilde f^\ell_* B}
    + c_\ell \log \alpha \int \tilde q^\ell 
    - c_\ell \int \tilde q^\ell + \frac {c_\ell^2}\alpha
    \int \tilde f^\ell  \tilde f^\ell_* B. 
      \end{split}
\end{equation*}
Since $\tilde q^\ell \le q$,  $\tilde f^\ell \ge f$,
and $c_\ell \to 1$, by dominated
convergence the first term on the right-hand-side converges to
$\int q \log (2q/ff_* B)$.
Since $\int q^\ell \to \int q$ and $\alpha \to 1$,
the second term tends to $0$, and the third converges to $Q(1)$.
Finally, since $\int q^{(3)} \zeta_0 < +\infty$,
$B$ is uniformly integrable with respect to $\tilde f^\ell
\tilde f^\ell_*$, therefore the last term converges to $Q^\pi(1)$.
\end{proof}

%%%%%%%%%%%%%%%%%%%%%%%%%%%%%%%%%%%%%%%%%%%%%%%%%%%%%%%%%%%%%%%%%%%%%%%%%%%%%%%%%%%%%%%%%%%%%%%%%%%%%%%%%%%%

\section{Large deviations for Kac model with
  canonical initial data}
\label{sez:6}

In this section we consider
the Kac model with canonical initial data, namely when
the initial velocities are i.i.d.\
sampled from a given probability $m$.
In view of the abstract Proposition \ref{t:ldtp}, the large deviation
principle for the pair empirical measure and flow
can be deduced from  the large deviation
principle of the Kac model with microcanonical initial data.

The canonical rate function is given by
\begin{equation}
  \label{eq:Ican}
  I(\pi,Q) = \inf_{(e,u)\in Z} \big( A(e,u) + I_{e,u} (\pi,Q) \big),
\end{equation}
where $A$, as defined in \eqref{cra}, is the rate function
relative to the sum of i.i.d.\ random variables
given by Cram\'er's theorem.

In order to compare this rate function with the one
in \cite{Le,He}, 
consider the dynamical function as in \eqref{5}, but without the
microcanonical constraint, namely 
\begin{equation}
  \label{Jfree}
  J(\pi,Q)\coloneqq 
    \int  \de Q^\pi \Big[ 
    \, \frac{\de Q\phantom{^\pi}}{\de Q^\pi} \log \frac{\de Q\phantom{^\pi}}{\de Q^\pi} -
    \Big( \frac{\de Q\phantom{^\pi}}{\de Q^\pi}  -1\Big)\Big].
\end{equation}
Then functional in \cite{Le,He} reads
$$\mc I(\pi,Q) = \Ent(\pi_0|m) + J(\pi,Q).$$
By Remark \ref{remark1}, for any $(\pi,Q)\in \ms S$ we have
$\mc I (\pi,Q) \le I(\pi,Q)$.
For some path $(\pi,Q)$ this inequality is strict because,
as discussed in detail in the next section,
$\mc I$ vanishes on  Lu and Wennberg solutions,
while $I$ is strictly positive.

\begin{theorem}
  \label{canonico}
  
  Let $m$ by a probability measure in $\bb R^d$ 
  and set
  $\mu^N = m^{\otimes N}$.
  If $m$ satisfies item (i)--(iii)
  in Assumption \ref{ass:2}
  then the family
  $\bb P^N_{\mu^N} \circ (\pi^N,Q^N)^{-1}$ satisfies a large deviation
  upper bound with good rate function $I: \ms S \to [0,+\infty]$,
  namely $I$ has compact level sets and for each closed $C\subset \ms S$
  \begin{equation}\label{upeqc}
    \varlimsup_{N\to +\infty} \frac 1N \log   \bb P^N_{\mu^N}
    \Big( (\pi^N,Q^N)\in C
    \Big) \le - \inf_{C} I.
  \end{equation}
  Moreover, if $m$ satisfies also
  condition {\rm (iv)} in Assumption \ref{ass:2}, then for
 each open  $O\subset \ms S$
  \begin{equation}\label{lbeqc}
    \varliminf_{N\to +\infty} \frac 1N \log   \bb P^N_{\mu^N}
    \Big( (\pi^N,Q^N)\in O
    \Big) \ge - \inf_{O\cap \hat{\ms S}} I.
  \end{equation}
\end{theorem}

\begin{proof}
  By the definition of the microcanonical ensemble $\nu^N_{e,u}$
  given below equation \eqref{sig}, we have 
  $$\bb P^N_{\mu^N} = \int p_N(\de (e,u)) \bb P^N_{\nu^N_{e,u}},$$
  where $p_N$ is the law of $\frac 1N \sum_i \bs \zeta(v_i)$ with
  $\bs v$ sampled according to $\mu^N$.
  By Cram\'er's theorem, as discussed before remark \ref{remark1},
  $p_N$ satisfies a large deviation principle with rate function $A$.
  The proof is thus essentially achieved by combining
  Theorem~\ref{upperbound} with the abstract Proposition~\ref{t:ldtp}.
  However, since in the large deviation result with microcanonical
  initial data the upper and lower bound rate function may differ, we need
  a replacement for Step 2 in the  the proof of Proposition~\ref{t:ldtp}.

  \noindent
  {\it Upper bound}.
  The argument in Step 4 in the proof of Proposition~\ref{t:ldtp}
  applies, provided we show that the map
  $Z\times \ms S\ni (e,u,\pi,Q) \mapsto
  I_{e,u}(\pi,Q)$
  is lower semicontinuous.

  Recall the set $\ms C_{e,u}$ defined in \eqref{ceut},
  and let $\mf C$  be the subset of $Z\times \ms S$ defined by
  $$\mf C \coloneqq \{(e,u,\pi,Q):\, \pi \in \ms C_{e,u}\}.$$
  %$\Psi_{e,u}(\pi,Q)$ be the function equal to $1$
  %if $\pi\in \ms C_{e,u}$
  %and $+\infty$ otherwise.
  By the lower semicontinuity of the map $\pi \mapsto
  \pi(\zeta_0)$, and the continuity of the map
  $\pi \mapsto \pi(\zeta)$ when the energy of $\pi$ is uniformly bounded,
  we deduce that $\mf C$ is closed.
%  $\Psi$ is jointly  lower  semicontinous.
  By the variational representation \eqref{true-var}, this
  implies the joint lower semicontinuity of $J_{e,u}(\pi,Q)$.

  By Theorem~\ref{sld} and Step 2 in the proof of Proposition~\ref{t:ldtp},
  we also deduce the joint lower semicontinuity of $H_{e,u}(\pi_0)$,
  that conclude the proof.

  \noindent
  {\it Lower bound}.
  Fix $(\pi,Q)\in \hat {\ms S}$.
  By Step 3  in the proof of proposition \ref{t:ldtp},
  we deduce that for any open neighborhood $\mc N$ of $(\pi,Q)$
  we have
  $$\varliminf_{N\to +\infty} \frac 1N \log   \bb P^N_{\mu^N}
  \Big( (\pi^N,Q^N)\in \mc N
  \Big) \ge -  I(\pi,Q),$$
  that implies the statement.
\end{proof}

\section{Asymptotic probability of
  Lu and Wennberg solutions}
\label{sez:5}

We start by observing that
the balance equation \eqref{bal}
for a pair $(\pi,Q)$ with 
$Q=Q^\pi$ is equivalent to the statement that
$\pi$ is a weak solution \eqref{eq:hb}.
Recalling that the functional $J$,
as defined in \eqref{Jfree}, vanishes
if and only if $Q=Q^\pi$, then
we deduce that the zero level set of $J$ are
the weak solutions to the homogeneous Boltzmann equation
\eqref{eq:hb}.
As we next state, the zero level set of both the
functional $I_{e,u}$ and $I$ respectively defined
in \eqref{I}, \eqref{eq:Ican} is a singleton. As a consequence
the large deviation upper bound stated in theorems
\eqref{upperbound} and \eqref{canonico} implies the
convergence of the empirical measure
to the unique energy solution to the
homogeneous Boltzmann equation \eqref{eq:hb}
with an exponential bound on the error.
\begin{theorem}
  \label{th:sol+funz}
  $~$
  \begin{itemize}
  \item[(i)]
    $I_{e,u} (\pi,Q)= 0$
    if and only if 
    $\pi = f\de v$, $Q=Q^\pi$ and  
    $f$ is the unique energy conserving
    solution to the Cauchy problem associated to \eqref{eq:hb}
    with initial datum $\frac  {\de m_{\mathrlap{e,u}}}{\de v}$\ \ \ 
    as defined in \eqref{meu}.
  \item[(ii)]
    $I(\pi,Q)= 0$
    if and only if 
    $\pi = f\de v$,  $Q=Q^\pi$ and 
    $f$ is the unique energy conserving
    solution to the Cauchy problem associated to \eqref{eq:hb}
    with initial datum $\frac {\de m}{\de v}$.
  \end{itemize}
\end{theorem}
\begin{proof}
  We prove only the first statement.
  By definition of $I_{e,u}$ if $f$ is an energy conserving
  solution to the Cauchy problem associated to \eqref{eq:hb} with
  initial datum $\frac {\de m_{\mathrlap{e,u}}}{\de v}$,\ \  then $\pi = f\de v$
  and $Q=Q^\pi$ belong to the zero level set of $I_{e,u}$.
  To prove the converse, we observe that, 
  by the very definition \eqref{5}, $J_{e,u}(\pi,Q) = 0$
  implies that $Q=Q^\pi$ and $\pi_t(\zeta_0)\le e$
  for any $t\in [0,T]$. Since $H_{e,u}(\pi_0) = 0$
  implies that $\pi_0 = m_{e,u}$ we deduce
  $\pi_t = f_t \de v$ where $f$ is 
  a weak solution to the Cauchy problem associated to \eqref{eq:hb}
  with initial datum $\frac {\de m_{\mathrlap{e,u}}}{\de v}$\ \ \ and
  non increasing energy.
  Since for any weak solution to \eqref{eq:hb}
  the energy can not decrease in time
  (see \cite{Lu,MW}), $f_t$ is the unique
  energy conserving solution.
\end{proof}
Fix a non-decreasing piecewise constant, left-continuous profile
$\mathcal E:[0,T]\to \bb R_+$, with finite, non zero, number of jumps.
\begin{definition}
  A Lu and Wennberg solution to the Cauchy
  problem associated to the homogeneous Boltzmann equation with
  initial datum $f_0$ and 
  \emph{energy profile} $\mc E$ is a measurable function
  $f:[0,T]\times \bb R^d \to [0,+\infty)$ such that
    \begin{itemize}
    \item[(i)] the map $t\mapsto f_t(v)\de v
      \eqqcolon \pi_t$ in $C([0,T]; \ms P(\bb R^d))$;
    \item[(ii)] $f$ is a weak solution to the homogeneous Boltzmann equation;
    \item[(iii)] $\pi_t(\zeta_0)=\mc E(t)$, $t\in [0,T]$.
    \end{itemize}  
\end{definition}
Observe that 
for any $e\geq \mc E(T)$, for $\pi = f\de v$, with $f$ a Lu and Wennberg
solution,
$J_{e,u}(\pi, Q^\pi)=0$. Hence
\begin{equation*}
  I_{e,u}(\pi, Q^\pi)=\Ent(\pi_0|m_{e,u}) + \big[\gamma_0^*-\gamma_0\big]
  \big[e-\mathcal E(0)\big],
\end{equation*}
namely the Lu and Wennberg solutions contribute to the rate function
only at time zero. We remark that these pairs $(\pi, Q^\pi)$ do not
belong to the set $\hat {\ms S}$ for which the upper and lower bound
in Theorem~\ref{upperbound} is proven to match. In the next theorem we
will show they actually match also for a suitable class of Lu and Wennberg
solutions.
\begin{theorem}
  \label{ld-lws}
  Fix $(e,u)\in Z$
  and a sequence $(e_N, u_N)\to (e,u)$.
  For each energy profile $\mc E$ with
  $\mc E(T) < e$ and
  each $f_0$ with energy $\mc E(0)$,
  there exists a Lu
  and Wennberg solution $f$ with energy profile $\mc E$ such that for
  every open neighborhood $A$ of $(\pi,Q^\pi)$, $\pi = f \de v$, 
  \begin{equation}
    \label{eq:vli-lw}
    \varliminf_{N\to+\infty}\frac 1 N \log\bb P^N_{\nu^N_{e_N, u_N}}
    \Big(
    (\pi^N,Q^N)\in A \Big) \geq - I_{e,u}(\pi, Q^\pi).
  \end{equation}  
\end{theorem}
Observe that, by the upper bound in Theorem~\ref{upperbound}
\begin{equation*}
  \varlimsup_{N\to+\infty}\frac 1 N \log\bb P^N_{\nu^N_{e_N, u_N}}
  \Big(
  (\pi^N,Q^N)\in \bar A \Big) \geq - \inf_{\bar A}I_{e,u},
\end{equation*}  
which, together 
with \eqref{eq:vli-lw},
identifies the asymptotic probability of Lu and Wennberg solutions.

As in \cite{LuW}, the Lu and Wennberg solutions will be constructed as a
limit of a suitable sequence. In particular we will consider a
sequence $f^n$ which conserve the energy and such that
$t\mapsto f^n_t(v)\de v\in \ms P(\bb R^d)$ is continuous.
We start with the  static result. 
\begin{lemma}\label{lemmag}
  Consider $\rho\in \ms P(\bb R^d)$ such that $H_{e,u}(\rho)$ is finite and  $e_0\coloneqq  \rho(\zeta_0) < e$. Given $e_1\in (e_0, e]$ and   $n\in \bb N$, let $g_n=m_{n(e_1 - e_0), u}$ be  the exponential tilt of $m$ with energy $n(e_1-e_0)$ and momentum $u$. Set $\rho_n= (1-\frac 1 n)\rho + \frac 1 n g_n$, so that $\rho_n(\zeta_0)=e_1-e_0$, then 
  \begin{equation}
\lim_{n\to\infty} H_{e,u}(\rho_n)= H_{e,u}(\rho).
  \end{equation}  
\end{lemma}  

\begin{proof}
  By the lower semicontinuity of $H_{e,u}$, it is enough to show that
  $\varlimsup H_{e,u}(\rho_n)\leq H_{e,u}(\rho)$. By the convexity of
  $H_{e,u}$ and Jensen inequality
  \begin{equation*}
H_{e,u}(\rho_n) \leq \big (1-\frac 1 n \big)H_{e,u}(\rho) + \frac 1 n H_{e,u}(g_n).
\end{equation*}
  Let $\bs{\lambda}^n $ such that
  $$g_n(\de v)= \frac { \ee^{\bs{\lambda}_n \cdot\bs\zeta } m(\de v) }{m \big(\ee^{\bs{\lambda}_n \cdot\bs\zeta} \big)}= \frac{\ee^{(\bs{\lambda}_n -\bs{\gamma}(e,u))\cdot\bs\zeta }}{m_{e,u} \big(\ee^{(\bs{\lambda}_n -\bs{\gamma}(e,u))\cdot\bs\zeta }\big)} m_{e,u}(\de v),$$
where we used \eqref{meu}. Observe that $\lambda^n_0 \uparrow \gamma_0^*$  as $n\to +\infty$. Since $g^n$ has energy $n(e_1-e_0)$ we get
\begin{equation*}
\varlimsup_{n\to +\infty} \frac 1 n\Ent (g^n | m_{e,u}) \leq \varlimsup_{n\to +\infty} (\lambda_0^n -\gamma_0(e,u)) (e_1-e_0)=  (\gamma_0^* -\gamma_0(e,u)) (e_1-e_0),
\end{equation*}
which concludes the proof.
\end{proof}  

For any probability density $h$ with finite energy
let $\mc U_t(h), t\geq 0$, be the unique energy conserving solution
to the Cauchy problem associated to the homogeneous Boltzmann
equation with initial datum $h$.
In the following statement we collect the result on moment estimate
in \cite{MW,W}.
\begin{lemma}\label{lemma:mom}
  Let $h$ be a probability density on $\bb R^d$ with finite energy and
  entropy. Then
  \begin{itemize}
  \item [(i)] For each $p>2$ and $t>0$ there exists a real  $C>0$ depending only on $p$, $t$ and the initial energy, such that
    $$ \int \de v\, \mc U_t(h)(v) |v|^p\leq C.$$
  \item [(ii)] For each $p>2$, if $\int \de v h(v) |v|^p <+\infty$, then
    $$
    \sup_{t\in [0,T]}\int \de v\, \mc U_t(h)(v) |v|^p < +\infty. %Elmroth
    $$
  \end{itemize}
\end{lemma}

Fix an energy profile $\mathcal E:[0,T]\to \bb R_+$
and denote by $0\leq t_1 < ..<t_k < T$ the discontinuity set of $\mathcal E$.
Given $f_0$ with finite entropy and energy $\mc E(0)$,
let $h_0^n$ be a sequence weakly convergent to $f_0$ satisfying the
following requirements.
The energy of $h_0^n$ is independent on $n$ and equal to $\mc E(0)$,
its entropy converges the entropy of $f_0$, and 
it has finite ($n$-dependent) $p-$moment for some $p\ge 3$.
For $n\geq 1$ and $i=1,..,k$, set
$e_{n,i}=n k[\mc E(t_i^+)-\mc E(t_i)]$
and define $g^n_i$ as the density of the tilted probability
$m_{e_{n,i}, u}$.
Define
\begin{equation}\label{aprLW}
  f^n_t =\begin{cases}
 \big( 1-\frac 1 n \big)\mc U_t (h_0^n) +\frac 1 {n k}\sum_{i=1}^k
 g^n_i  & t\in [0, t_1]\\
 \big( 1-\frac {k-1} {n k} \big)\mc U_{t-t_1}
 \Big( h_1^n\Big) +\frac 1 {n k}\sum_{i=2}^k  g^n_i  & t\in (t_1, t_2]\\
  ... &...\\
  \big( 1-\frac 1 {n k} \big)\mc U_{t-t_{k-1}} \Big( h_{k-1}^n\Big) +
  \frac 1 {n k}  g^n_k  & t\in (t_{k-1}, t_k]\\
  U_{t-t_k} (h^n_k) & t\in (t_k, T],\\
  \end{cases}
\end{equation}  
where $h^n_i$ are recursively defined so that $t\mapsto f^n_t(v)\de v$
is continuous, namely
$$h^n_i =  \tfrac 1 {1 -\frac{k-i}{ n k}}\Big[f^n_{t_i} -\frac 1 {n k}\sum_{j=i+1}^k g^n_j\Big].$$
Let also $q^n_t(v,v_*,\omega)$ be such that, for
$t\in (t_i, t_{i+1}]$,
  \begin{equation*}
q^n_t(v,v_*,\omega) =\Big( 1-\frac {k-i} {n k} \Big)\mc U_{t-t_i}\big( h_i^n\big)(v)\mc U_{t-t_i}\big( h_i^n\big)(v_*) B(v, v_*, \omega).
\end{equation*}
Here $i=0,..,k$, with $t_0=0$ and $t_{k+1}=T$.  Observe that, by
construction, the pair $(\pi^n, Q^n)$ with densities $(f^n, q^n)$
satisfies the balance equation \eqref{bal}.  Furthermore, by definition of
$h_0^n$ and item (ii) in Lemma \ref{lemma:mom}, for each $n$ the pair
$(\pi^n, Q^n)\in \hat{\ms S}$.
  
\begin{lemma}
  The sequence $\{(\pi^n, Q^n)\}$ is relatively compact in $\ms S$.
  Any cluster point $(\pi, Q)$ is such that $Q=Q^\pi$, 
  $\pi=f\de v$, where $f$ is
  a Lu and Wennberg solution with initial datum $f_0$ and
  energy profile $\mc E$.
  Moreover
  \begin{equation}\label{ILW}
    \lim_{n\to\infty}I_{e,u}(\pi^n, Q^n)=H_{e,u}(f_0\de v).
  \end{equation}
\end{lemma}

\begin{proof}
  We start by proving \eqref{ILW}. Observe that
  $\int \de v\, h^n_i \zeta_0 =\mc E(t_i^+)$, for $i,1..,k$. Then by
  Lemma \ref{lemmag} and Jensen inequality,
  \begin{equation*}\begin{split}
      \lim_{n\to\infty} H_{e,u}(\pi^n_0) = &
      \Ent(f_0\de v|m_{e,u})+ (\gamma_0^* -\gamma(e,u))\Big[e-\mc E(T) +
      \sum_{i=1}^k (\mc E(t_i^+)-\mc E(t_i))\Big]\\
        = & H_{e,u}(f_0 \de v).
    \end{split}\end{equation*}
    We now show that
    \begin{equation}\label{Jvan}
\lim_{n\to\infty }J_{e,u}(\pi^n, Q^n)=0.
\end{equation}
By definition, for 
$t\in (t_i, t_{i+1}]$, $i=0,..,n$, we have
\begin{equation*}
f^n_t \geq \big( 1- \frac {k-i} n \big)\mc U_{t-t_i}(h^n_i).
\end{equation*}  
Hence the the contribution to $J_{e,u}(\pi^n, Q^n)$ in the time window $(t_i, t_{i+1}]$ is bounded by
  \begin{equation*}\begin{split}
\int_{t_i}^{t_{i+1}}\de t\int \de v \de v_* \de \omega \, \big\{ q_t^n(v, v_*, \omega)\log \big( 1- \frac {k-i} n \big)^{-1} \\-q_t^n(v, v_*, \omega) + f_t^n(v)f_t^n(v_*) B(v,v_*,\omega)\big\}.
  \end{split}\end{equation*}  
Since the energy of $f_t^n$ is $\mc E(T)$, the mass of $q^n$ is
bounded uniformly in $n$, therefore the first term vanishes as
$n\to\infty$.  The same bound, together with the fact that the energy
of $\frac 1 n g^n_i $ is bounded uniformly in $n$ and
$B\leq C(1 + |v|+ |v_*|)$, implies that the second line also
vanishes. Hence \eqref{Jvan} follows.

Equation \eqref{ILW} and the goodness of $I_{e,u}$ imply that the
sequence $(\pi^n, Q^n)$ is relatively compact. Let $(\pi, Q)$ a
cluster point. By the lower semicontinuity of $J_{e,u}$ and
\eqref{Jvan} we deduce that $J_{e,u}(\pi, Q)=0$, hence
$Q=Q^\pi$, $\pi = f \de v$,
where $f$ is a solution to the Cauchy problem
associated to \eqref{eq:hb} and initial datum $f_0$.
It remains to show that $f$ has energy profile $\mc E$.
For any $i=0,..,k$, the energy of $h_i^n$ is uniformly bounded.  Fix
$i$ and $t\in (t_i, t_{i+1}]$. By item (i) in \ref{lemma:mom}, the
$p$-moment of $\mc U_{t-t_i}(h_i^n)$, $p>2$, is bounded uniformly in
$n$, therefore $\zeta_0$ is uniformly integrable with respect to
$\mc U_{t-t_i}(h_i^n)$, then
$$
\int \de v\, f_t \zeta_0 =\lim_{n\to+\infty} \int \de v\, \mc U_{t-t_i}(h_i^n)\,\zeta_0 =\mathcal E(T) -\sum_{j=i+1 }^k \big (\mc E(t_j^+)-\mc E(t_j)  \big) =\mc E(t).
$$
\end{proof}

\begin{theorem}
  \label{ld-lws-can}
  Let $m$ be a probability measure satisfying Assumption \ref{ass:2},
  and set $\mu^N \coloneqq m^{\otimes N}$.
  For each energy profile $\mathcal E$ with
  $\mc E(0) = m(\zeta_0)$
  there exists a Lu
  and Wennberg solution $f$ with $f_0 = m$ and energy profile $\mc E$ such
  that for  every open neighborhood $A$ of $(\pi,Q^\pi)$, $\pi = f \de v$, 
  \begin{equation}
    \label{eq:vli-lw-can}
    \varliminf_{N\to+\infty}\frac 1 N \log\bb P^N_{\mu^N}
    \Big(
    (\pi^N,Q^N)\in A \Big) \geq - I(\pi, Q^\pi) = \gamma_0^*
    (\mc E(T) - \mc E(0)).
  \end{equation}
\end{theorem}
\begin{proof}
  The proof of the inequality in \eqref{eq:vli-lw-can}
  follows the same arguments of the proof of Theorem~\ref{ld-lws}.
  We here discuss
  the equality.
  Since $\pi=f \de v $ is a weak solution to \eqref{eq:hb},
  $J_{e,u}(\pi,Q) = 0$ if $e\ge \mc E(T)$, otherwise
  is infinity.
  Then,
  by definition \eqref{eq:Ican} and Theorem~\ref{sld}
  we have that
  $$I(\pi,Q) = \inf_{e\ge \mc E(T)} ( A(e,u)+ \Ent(m|m_{e,u}) +
  (\gamma_0^*-\gamma_0(e,u)(e - \mc E(0))$$
  where $u = m(\zeta)$.
  The supremum in the  definition \eqref{cra} of $A_{e,u}$ is achieved
  in $\bs \gamma = \bs \gamma(e,u)$. By definition of
  the relative entropy
  $$\Ent(m|m_{e,u}) = -\bs \gamma(e,u) \cdot m(\bs \zeta) +
  \log m( \ee^{\bs \gamma(e,u) \cdot \bs \zeta}).$$
  Then, by direct computation,
  $I(\pi,Q) = \inf_{e\ge \mc E(T)} \gamma_0^*(e-\mc E(0)) =
  \gamma_0^*(\mc E(T)-\mc E(0))$.
\end{proof}
\appendix

\section{}
\label{appendixa}

  It is sufficient to 
  prove  that   $Q^\delta\big(\big[\log 1/B\big]^+\big)$
  converges to
  $Q\big(\big[\log 1/B\big]^+\big)$
  as $\delta \to 0$,
  since 
  the result for the negative part easily follows from the
  fact that $\left| \big[\log 1/B\big]^-\right|$ is sublinear in $|v-v_*|$,
  and $(\pi,Q)\in \hat {\ms S}$.
  We indicate with $g^1_\delta$ the Gaussian kernel in one dimension, 
  and note that
  $$(g_\delta \otimes g_\delta \otimes \opid) * \bigg[\log
  \frac 1{2B} \bigg]^+
  (v,v_*,\omega)
  = \int_{\bb R} g^1_\delta( w\cdot \omega -y )
  \bigg[\log \frac 1{|y|}\bigg]^+\de y,$$
  where $w = (v-v_*)/\sqrt{2}$.
  We now prove that there exist some constants $c_1,c_2>0$ such
  that
  $$\int_{\bb R} g^1_\delta( x-y )
  \bigg[\log \frac 1{|y|}\bigg]^+\de y \le c_1
  \bigg[\log \frac 1{|x|} \bigg]^+ + c_2,$$
  which implies that
  $$(g_\delta \otimes g_\delta \otimes \opid) *
  \bigg[\log \frac 1B\bigg]^+ (v,v_*,\omega)
  \le c_1  \bigg[\log \frac 1B\bigg]^+ (v,v_*,\omega) + c_2.$$
  Using this fact and that 
  $Q\big(\big[\log 1/B\big]^+\big) < + \infty$, 
  we achieves the convergence result 
  by using Fubini-Tonelli theorem and dominate convergence.

  % We rewirte $[\log 1/|y|]^+ = \log 1/(|y|\wedge 1)$ and denote by
  We denote by 
  $z$ a standard Gaussian stochastic variable
  and note that  
  $$
  \int_{\bb R} g^1_\delta( x-y )
  \bigg[\log \frac 1{|y|}\bigg]^+ \de y  =
  \bb E \bigg(  \bigg[\log \frac 1{|x-\delta z|}\bigg]^+\bigg)
  \le \log \frac 1{\delta} +   \bb E \bigg(
  \bigg[\log \frac 1{|x/\delta-z|}\bigg]^+\bigg).
  $$
  Since $\big[\log 1/|y|\big]^+$ is summable,
  by the Young's inequality the second term is uniformly bounded, 
  so that,
  if $|x|\le \sqrt{\delta}$ we have
  $$\bb E \bigg(  \bigg[\log \frac 1{|x-\delta z|}\bigg]^+\bigg)
  \le  2 \log \frac 1{|x|} + c.$$
  To handle the case $|x|\ge \sqrt{\delta}$, we use
  the Jensen inequality:
  $$\bb E \bigg(  \bigg[\log \frac 1{|x-\delta z|}\bigg]^+\bigg)
  = 
  2\log \ee^ {\bb E \big( \big[ \log 1/\sqrt{|x-\delta z|}\big]^+\big)}
  \le
  2 \log \bb E \left( \frac 1{\sqrt {|x-\delta z| \wedge 1}}\right)
  $$
  We estimate 
  $$\bb E \left( \frac 1{\sqrt {|x-\delta z| \wedge 1}}\right) =
  \int_{\bb R} g^1_\delta(y)
  \frac 1{\sqrt {|x-y| \wedge 1}}\de y
  $$
  by noticing that 
  in the region  $|y|<|x|/2$ or $|y|>2|x|$
  we have $1/\sqrt {|x-y| \wedge 1}\le  \sqrt {2} / \sqrt{|x|\wedge 1}$.
  Therefore
  $$\bb E \left( \frac 1{\sqrt {|x-\delta z| \wedge 1}}\right)
  \le c \frac 1{\sqrt{|x|\wedge 1}} + g^1_\delta(|x|/2) \int_{|x|/2}^{2|x|}
  \frac 1{\sqrt {|x-y| \wedge 1}}\de y.$$
  We conclude the proof observing that
  the last term is estimate by $c\ee^{-1/8\delta} (1+1/\delta)$, which is
  uniformly bounded in $\delta$.


\begin{thebibliography}{99}

%\bibitem{ADPZ} Adams S., Dirr N.,  Peletier M. A., Zimmer J.;
%  {\it From a large-deviations principle to the Wasserstein gradient flow: a new micro-macro passage},
%  Communications in Mathematical Physics 307, 3, 791--815, 2011.
  
%\bibitem{AGS} Ambrosio L., Gigli N., Savar\'e G.; {\it Gradient flows in metric spaces and in the spaces of probabilitymeasures},  Lectures in Mathematics, ETH Zurich, Birkh\"auser, 2005.

%\bibitem{BBB} Basile G., Benedetto D., Bertini, L.; {\it A gradient flow approach to linear Boltzmann equation}, Ann. Sc. Norm. Super. Pisa Cl. Sci. (5) Vol. XXI (2020), 955-987



\bibitem{BBBC} Basile G., Benedetto D., Bertini L., Caglioti C.;
  {\it Large deviations for a binary collision model: energy evaporation},
  Mathematics in Engineering, 5(1): 1–-12 
  DOI:10.3934/mine.2023001 (2023)

  

\bibitem{BBBO} Basile G., Benedetto D., Bertini L., Orrieri C.;
  {\it Large Deviations for Kac-Like Walks},
  J. Stat. Phys. 184, 10 https://doi.org/10.1007/s10955-021-02794-2
  (2021)
  
%\bibitem{BB} Basile G., Bertini L.; {\it Donsker-Varadhan asymptotics for degenerate jump Markov processes}, ALEA, Lat. Am. J. Probab. Math. Stat. 12 (1), 1–34 (2015).  
  

%\bibitem{BG} Ben Arous Guionnet
  
%\bibitem{BGSS1} T. Bodineau, I. Gallagher, L. Saint-Raymond and S. Simonella;
%  {\it Fluctuation Theory in the
%  Boltzmann–Grad Limit}, J. Stat. Phys. 180, 873-895, 2020.
  
\bibitem{BGSS2}
  T. Bodineau, I. Gallagher, L. Saint–Raymond and S. Simonella;
  {\it Statistical dynamics of a hard sphere gas: fluctuating Boltzmann
    equation and large deviations}, preprint, arXiv:2008.10403 (2020).
  
  % \bibitem{Bou}  F. Bouchet; {\it Is the Boltzmann Equation Reversible?
  %   A Large Deviation Perspective on the Irreversibility Paradox},
  %   J. Stat. Phys., 2020.
  
  
  % \bibitem{Bo} A. V. Bobylev;
  %   {\it Moment inequalities for the Boltzmann equation and applications to
  %   spatially homogeneous problems}
  %   J. Statist. Phys. 88(5/6):1183--1214 (1997).
  
  
  
  % \bibitem{DG} Dawson D. A., G\"artner J.;{\it Large deviations from the McKean-Vlasov limit for weakly interacting diffusions}, Stochastics 20,
  % 4, 247â308, 1987.

\bibitem{CCL-RLV}
  Carlen E. A.,  Carvalho M. C.,  Le-Roux J.,  Loss M., Villani C.;
  \emph{Entropy and chaos in the Kac model} Kinet. Relat. Models 3 
  no. 1, 85--122 (2010).

\bibitem{DZ} Dembo A., Zeitouni O.;
  \emph{Large Deviations Techniques and Applications}
  volume 38 of Applications of Mathematics,
  Springer-Verlag, New York,
  ISBN 0-387-98406-2,
  second edition (1998). 
    
%\bibitem{DV} Donsker  M. D.,  Varadhan  S. R. S.;
%  {\it Asymptotic evaluation of certain Markov
%    process expectations for large time. I. II.}, Comm. Pure Appl. Math.
%  28, 1--47, ibid. 28,  279--301, 1975.

%\bibitem{Er} Erbar M.; {\it Gradient flows of the entropy for jump processes},Ann. Inst. H. PoincarÃ© Probab. Statist. 50, 3, 920-945, 2014.

\bibitem{Er2} Erbar M.;
  {\it A gradient flow approach to the Boltzmann equation},
  arXiv:1603.00540v2 (2017).   


%  NON C'E'
%\bibitem{Fe} Feller W.;
%  {\it An introduction to probability theory and its applications},
%  Vol. II,   John Wiley \& Sons, Inc., New York-London-Sydney (1971).

\bibitem{He} Heydecker D.;
  {\it
    Large Deviations of Kac's Conservative Particle System and Energy Non-Conserving Solutions to the Boltzmann Equation: A Counterexample to the Predicted Rate Function}  arXiv:2103.14550 (2021).
  
%  NON C'E'
%\bibitem{K} Kac M.; {\it Foundations of kinetic theory}, Proc. 3rd
%  Berkeley Symp. Math. Stat. Prob., J. Neyman, ed. Univ. of California,
%  vol. 3, 171--197 (1956).
  
    
\bibitem{KR}  Kim S.S., K. Ramanan K.; 
  {\it A conditional limit theorem for high-dimensional $\ell^p$ spheres},
  J. Appl. Probab. 55 no. 4
  1060--1077  (2018).
  
\bibitem{KL} Kipnis C., Landim C.;  {\it Scaling Limits of Interacting Particle Systems}, volume 320 of Grundlehren der
  Mathematischen Wissenschaften [Fundamental Principles of
  Mathematical Sciences], Springer-Verlag, Berlin (1999). 
    
\bibitem{Le}
  L\'eonard, C.; {\it On large deviations for
    particle systems associated with spatially
    homogeneous Boltzmann type equations} Probab. Th. Rel. Fields 101,
  1--44 
  https://doi.org/10.1007/BF01192194 (1995).

\bibitem{Lu} Lu X.;
  {\it Conservation of energy, entropy
    identity, and local stability for the spatially homogeneous
    Boltzmann equation} J. Statist. Phys. 96(3), 765--796
  (1999).
  
\bibitem{LuW} Lu X., Wennberg B.,
  {\it Solutions with increasing
    energy for the spatially homogeneous Boltzmann equation}
  Nonlinear Analysis: Real World Applications
  (3) 2 pp. 243-–258 https://doi.org/10.1016/S1468-1218(01)00026-8 (2002)


  
 %
%\bibitem{Ma} Maas J.; {\it Gradient flows of the entropy for finite Markov chains}, J. Funct. Anal. 261, 2250â2292, 2011.
%
  
% \bibitem{Mar} Mariani M.; {\it A $\Gamma$-convergence approach to large deviations}, Ann. Sc. Norm. Super. Pisa Cl. Sci. 18, 951--976, 2018.


  
% \bibitem{MPR} Mielke A.,  Peletier M. A., Renger  D. R. M.; {\it On the relation between gradient flows and the large-deviation principle, with applications to Markov chains and diffusion}, Potential Analysis 41, 4, 1293--1327, 2104. 


\bibitem{MW} Mischler S., Wennberg B.;
  {\it On the spatially homogeneous Boltzmann equation}
  Ann. de l'I.H.P. Analyse non linéaire, Tome 16 (1999) no. 4, pp. 467--501. 
  
\bibitem{W} Wennberg B.;
  \emph{Entropy dissipation and moment production for the Boltzmann equation}
  J. Statist. Phys. 86(5/6) 105--1066 (1997).

\bibitem{Pe} Petrov V.V.;
  \emph{Sums of independent random variables}
  Springer-Verlag, New York, ISBN 978-3-540-06635-4 (1975).


\bibitem{Nam} Nam K.; {\it  Large deviations
    and localization of the microcanonical
ensembles given by multiple constraints}
Ann. Probab. 48(5) 2525--2564
DOI: 10.1214/20-AOP1430 (2020).

\bibitem{Re}  Rezakhanlou F.; {\it Large deviations from a kinetic limit}
  Annals of Prob. 26(3) 1259–-1340 (1998).



% \Bibitem{Sn} Sznitman, A.S.; {\it Topics in propagation of chaos}, in Hennequin PL. (eds) Ecole d'Et\'e de Probabilit\'es de Saint-Flour XIX - 1989, Lect. Notes Math., vol 1464, Springer Berlin Heidelberg, 1991  

%   \bibitem{Ta} Tanaka H.; {\it Fluctuation Theory for Kac's One-Dimensional Model of Maxwellian Molecules}, The Indian Journal of Statistics, Series A (1961-2002), Vol. 44, No. 1 (1982), pp. 23-46

%   \bibitem{Uc1} Uchiyama, K.; {\it A fluctuation problem associated with the Boltzmann equation for a gas of molecules with a cutoff potential}, Japan J. Math. 9, 27-53 (1983).

%     \bibitem{Uc2} Uchiyama, K.; {\it Fluctuations in a Markovian system of pairwise interacting particles}, Probab. Theory Relat. Fields 79, 289-302 (1988)

%       Blp,  DGP


\end{thebibliography}
\end{document}